\documentclass{article}

\usepackage{arxiv}
\usepackage[T1]{fontenc}
\usepackage{graphicx}
\usepackage{color}

\usepackage{amsmath,amsfonts,amsthm}
\usepackage{graphicx}
\usepackage{appendix}
\usepackage{subfigure}
\usepackage{bm}
\usepackage{booktabs}
\usepackage{amssymb}
\usepackage{tikz}
\usepackage{url}
\usepackage{setspace}
\usepackage[pdftex,bookmarksnumbered,bookmarksopen,
colorlinks,citecolor=blue,linkcolor=blue,urlcolor=blue]{hyperref}
\usepackage{framed}
\usepackage{xcolor}
\usepackage{soul}
\usepackage{longtable}

\usepackage{times}
\usepackage{enumitem}
\usepackage{varwidth}
\usepackage{graphicx}
\usepackage{wrapfig}
\usepackage{enumerate}
\usepackage{caption}

\usepackage{floatrow}
\newfloatcommand{capbtabbox}{table}[][\FBwidth]

\usepackage{bbm}
\newtheorem{theorem}{Theorem}[section]
\newtheorem{lemma}[theorem]{Lemma}
\newtheorem{proposition}[theorem]{Proposition}
\newtheorem{example}{Example}

\newtheorem{remark}{Remark}

\newtheorem{assumption}{Assumption}
\newtheorem{definition}{Definition}

\newcommand{\x}{\bm{x}}

\newcommand{\C}{\mathbf{C}}

\newcommand{\G}{\mathbf{G}}

\newcommand{\I}{\mathbb{I}}

\newcommand{\R}{\mathbb{R}}
\renewcommand{\S}{\mathbf{S}}

\newcommand{\bbZ}{\mathbb{Z}}

\newcommand{\red}[1]{{\color{red}#1}}
\newcommand{\purple}[1]{}
\newcommand{\blue}[1]{}

\newcommand{\cA}{\mathcal{A}}

\newcommand{\cC}{\mathcal{C}}
\newcommand{\cD}{\mathcal{D}}

\newcommand{\cG}{\mathcal{G}}
\newcommand{\cH}{\mathcal{H}}
\newcommand{\cI}{\mathcal{I}}

\newcommand{\cS}{\mathcal{S}}
\newcommand{\cT}{\mathcal{T}}
\newcommand{\cU}{\mathcal{U}}
\newcommand{\cV}{\mathcal{V}}
\newcommand{\cN}{\mathcal{N}}

\newcommand{\cX}{\mathcal{X}}

\newcommand{\bbE}{\mathbb{E}}

\newcommand{\bbS}{\mathbb{S}}

\newcommand{\OPT}{\mathsf{OPT}}
\newcommand{\BR}{\text{BR}}

\DeclareMathOperator*{\argmin}{argmin}

\newenvironment{proofoutline}{\noindent{\emph{Proof Sketch. }}}{\hfill$\square$\medskip}
\begin{document}
\title{Subsidy design for better social outcomes}
%
%

\author{\textbf{Maria-Florina Balcan}\\
\texttt{ninamf@cs.cmu.edu} \\
        Carnegie Mellon University\\
        \and
		\textbf{Matteo Pozzi}\\
		\texttt{mpozzi@cmu.edu}\\
        Carnegie Mellon University
        \and
        \textbf{Dravyansh Sharma}\\
        \texttt{dravy@ttic.edu} \\
        Toyota Technological Institute at Chicago}
%
%
%

\maketitle              
\begin{abstract} 
    Overcoming the impact of selfish behavior of rational players in multiagent systems is a fundamental problem in game theory. Without any intervention from a central agent, strategic users  take actions in order to maximize their personal utility, which can lead to extremely inefficient overall system performance, often indicated by a high Price of Anarchy. Recent work \cite{lin2021multi} {investigated and formalized} yet another undesirable behavior of rational agents, that of avoiding freely available information about the game for selfish reasons, leading to worse social outcomes. A central planner can significantly mitigate these issues by injecting a subsidy to reduce certain costs associated with the system and obtain net gains in the system performance. Crucially, the planner needs to determine how to allocate this subsidy effectively.

We formally show that designing subsidies that perfectly optimize the social good, in terms of minimizing the Price of Anarchy or preventing the information avoidance behavior, is computationally hard under standard complexity theoretic assumptions. On the positive side, we  show that we can learn provably good values of subsidy in repeated games coming from the same domain. This data-driven subsidy design approach avoids solving computationally hard problems for unseen games by learning over polynomially many games. We also show that optimal subsidy can be learned with no-regret given an online sequence of games, under mild assumptions on the cost matrix. Our study focuses on two distinct games: a Bayesian extension of the well-studied fair cost-sharing game, and a component maintenance game with engineering applications. 

\end{abstract}

\section{Introduction}

Multiagent systems often need a central agent to intervene and avoid harmful consequences of selfish decisions of individual agents. Subsidy is a form of positive intervention where the central agent reduces the cost of some actions in the system, with the goal of leading agents to better social outcomes. The amount of subsidy available to the central agent is typically scarce and therefore it is crucial to optimize its allocation. Prior research has addressed this by designing subsidy schemes that approximately optimize the Price of Anarchy \cite{buchbinder2008non,ferguson2021effectiveness}. In this work we extend the study of subsidy design for games in three different directions: we study the impact of subsidy beyond the Price of Anarchy objective and show its usefulness in preventing an undesirable information avoidance behavior, we establish formal hardness results for designing optimal subsidy schemes, and  show how an alternative novel data-driven approach can be used to design subsidy by exploiting historical data or related games.

Concretely, critical infrastructure maintenance typically involves joint responsibility shared among multiple stakeholders, and  failure of coordination can lead to disastrous consequences. {For example, different segments of a large road network are typically managed by different civil authorities. 
With an increasingly connected physical and digital world, it is a major challenge to coordinate large-scale systems consisting of several disjointly owned components. As a result, one crucially needs a {\it central planner}---who has the ability to allocate shared resources to avoid major catastrophic failures---to ensure  smooth operation of the overall system. In the above examples, this central agent could be a government department with appropriate jurisdiction. Ideally, the central agent would ensure a judicious use of the common resources which are to be allocated to appropriate stakeholders to incentivize them to do their part. Identifying the optimal resource allocations can be hard, but 
very often the central agent manages multiple similar systems or has access to relevant historical data. Could one take advantage of this data availability to improve the allocation?

Alternatively, in modern market systems, consumers have several options and would prefer to select options that meet their needs at the smallest cost. Often the consumer commits to an action based on expected long-term costs. For example, people buy health insurance plans and effectively  share the cost of healthcare with other subscribers to the plan.  If they could accurately estimate their need for medical services, people who determine they would not need any expensive medical procedures would opt out and drive up the insurance costs. An intervention by the governement to reduce the cost of health insurance plans could ensure that these people still opt in and the system is robust to the additional information regarding need for medical services. This intervention to guide the market is expensive, and also needs a careful allocation. As with infrastructure maintenance, optimal allocation is hard and one would like to use historical market data to guide it.

\paragraph{Games studied.}  
We consider a component maintenance game  based on \cite{lin2021multi}. We model the common responsibility as disjoint {\it components} maintained by individual agents or stakeholders, each having a binary state denoting whether the component is functioning or broken. For simplicity, we will assume known (prior) probabilities which govern whether the component will work, which are a common knowledge among all agents. The overall {\it system} is also assumed to have a binary state, given by some boolean function of the component states. For example, the overall system consisting of five components might function only if all components are functioning, or it might function if at least one of the first three components is functioning and at least one of the next two components is functioning. The components are repairable, and the agent maintaining it can choose to repair their component at some personal cost, or choose to do nothing. If an agent repairs
their component, then it is assumed to be guaranteed to function. In addition to the repair cost which is only charged to agents that undertake repair, all agents are assumed to experience a large cost if the overall system fails (see Section \ref{sec:prelims} for formal details of the game). Furthermore, any individual component may be publically inspected (without any inspection cost) to determine whether it is actually working or broken. In this case, the agents may have a tendency to avoid having the state of their component inspected and revealed to everyone to avoid increased personal cost at equilibrium. {In our model, it is the agents that could inspect a component but do not control the information about the inspected state, i.e.\ the state of an inspected component is revealed to all agents irrespective of who owns the component.}

We also study a Bayesian extension of the classical cost-sharing game. We are given a collection of actions, each associated with some cost and  different actions being available to a fixed subset of agents. All agents that select a fixed action from the collection share the cost of that action. For example, if the actions are commute options like bus, train or car and agents are commuters, the cost of running the bus is shared by its users. In the Bayesian extension, there is a prior over of the action costs and agents choose actions based on the mean cost under the distribution. The true  cost of some action may be inspected and the information revealed to all the agents.

\paragraph{Summary of contributions.} We formalize the problem of effective resource allocation by a central agent to improve the performance of a multi-agent system by subsidizing certain costs associated with the system.  
We consider  distinct objectives that the  central agent might have,
\begin{itemize}
    \item[(a)] to use subsidy to reduce the price of anarchy for the system, i.e.\ to ensure that the harmful effect of the selfish behavior and lack of coordination and   of the agents on the social cost is minimized,
\item[(b)] to ensure that the value of information {(measured as change in agent's cost at equilibrium before and after inspection)} about the state of an inspected component is non-negative for all agents.
\end{itemize} 

\noindent We show that the system can perform poorly on each of the above objectives in the absence of any subsidy. The goal of the central agent therefore is to determine the smallest subsidy budget needed to ensure that one of the above objectives is met. We will show (Section \ref{sec:opt-subs-2-series}) that  this calculation of optimal budget by the central agent can be done exactly for a simple small system, and the optimal subsidy allocation can be different for different objectives. We further show that in contrast
it is computationally hard to do so in more general systems under standard complexity theoretic assumptions. The computational hardness results (located in Section \ref{sec:hardness}) for computing the optimal subsidy hold for both the above objectives, for both component maintenance and cost-sharing games.  

On the positive side, if the central agent has access to data about multiple games, we show (in Section \ref{sec:data-driven}) that a good value of the subsidy and allocation can be achieved with a polynomial number of game samples coming from an arbitrary game distribution. Moreover, if the games happen sequentially, the agent can perform nearly as well as the best subsidy scheme in hindsight, under very mild assumptions on the adversarial sequence of games. Since designing an optimal subsidy scheme is computationally hard, we would like to avoid having to solve the problem too many times in repeated games from the same domain. If we have access to similar games (e.g.\ infrastructure projects in similar counties), {the subsidy design problem is still  hard}, but we can potentially avoid a large number of repeated intensive computations. Moreover, the central agent may need to decide the value of subsidy on a new unseen game instance without observing the relevant parameters for this ``test'' game, for example the prior distribution on component failure. We obtain polynomial bounds on the sample complexity of the number of sample games needed to learn the optimal subsidy, which imply that we need to optimize the subsidy only for polynomially many games and can use the ``learned'' subsidy scheme on further game instances. 
We also obtain no-regret guarantees on learning the subsidy parameter in an online sequence of games, under a mild smoothness assumption on the repair costs. \purple{This could be useful for example in studying potential failure in a communication network with dynamically changing nodes/components.}

While Lin et al.\ \cite{lin2021multi} introduces the component maintenance game and expose the challenge of information avoidance in small systems with a constant number of agents, we study general  systems with an arbitrary number of components and consider additional relevant objectives from the perspective of a central agent that can provide subsidy to some agents to reduce the repair cost of their component. While the use of subsidy has been studied in the context of cost-sharing games \cite{buchbinder2008non}, we study an interesting Bayesian extension where agents may experience negative Value of Information, establish new formal hardness results for optimal subsidy design and an a data-driven approach for overcoming the computational hardness in subsidy allocation.

Our main tool for showing the above positive results is to employ a recently introduced paradigm for beyond worst-case analysis called {\it data-driven algorithm design} \cite{gupta2016pac,balcan2020data}. Unlike traditional analysis, where one gives exact or approximate performance bounds applicable for worst-case instances, this paradigm focuses on ``typical'' problem instances that one actually encounters. This is similar to average-case analysis, but instead of a uniform distribution over problem instances, any arbitrary (fixed but unknown) distribution over the problem instances is allowed.

\subsection{Related Work}

{\bf Component maintenance games}. {Management of engineering systems often involves maintaining multiple components arranged in some scheme that govern the overall functionality of the system; these components are controlled by different agents that make decisions under uncertainty} \cite{malings2016value,memarzadeh2016value,pozzi2020information}. Often this involves careful planning and resource allocation by a central agent whose goal is ensuring that the overall cost to the agents is small, and that the agents make good use of any available information about the component states \cite{raiffa1961applied}. The central agent could design incentives or subsidy to be given to specific agents for improving the overall system. The role of   subsidy and taxation has been studied in both cooperative and non-cooperative games \cite{buchbinder2008non,meir2010minimal,zick2013taxation}. Typically, a central agent designs a subsidy (or taxation) scheme, which effectively alters the game parameters by changing the costs/profits of the agents, to minimize some objective like the Price of Anarchy. We examine optimization of novel objectives in addition to previously analyzed ones in the context of component maintenance games, and demonstrate a first application of a learning-theoretic lens to overcome the worst-case computational hardness of subsidy design.

{\bf Cost-sharing games}.
Cost-sharing game is a classical game in algorithmic game theory. Several variants of this game have been studied in the literature, including a set cover version \cite{buchbinder2008non,meir2010minimal} which we study here, and multicast game where network users connect to a source by paying for a route to the source and sharing cost \cite{chekuri2006non,anshelevich2008price,charikar2008online,balcan2009price}. Prior work has shown that subsidy is crucial for the former, while best response dynamics are sufficient to obtain a small Price of Anarchy for the latter. While \cite{buchbinder2008non} propose a primal-dual approach for approximately minimizing PoA under subsidy, we complement their results by establishing NP-hardness of exact optimization and further study the information avoidance phenomenon in a Bayesian extension of the game.

{\bf Value of Information}.
Information avoidance has been studied extensively in behavioral sciences, economics, psychology and public health \cite{sharot2020people,golman2017information,sweeny2010information,ajekigbe1991fear}. People can decide to seek or avoid certain information based on their goals, and strategic considerations can cause agents to ignore free and useful information. In the component maintenance games that we study, the information  corresponds to the true state of some {\it inspected} component. Agents can choose to not seek this information \purple{(about their own component, or that of another agent)} even if it is freely available, if it could make them worse off (e.g.\ increase the need to repair their component) even at the cost of making the overall system more likely to fail. \cite{lin2021multi} demonstrates this phenomenon in several multi-agent network systems. \cite{lin2021multi} focuses on component inspection and  the value of information metric, and the results apply to games with a small constant number of agents.  In contrast, we study a broader variety of metrics of interest to the central agent, provide formal hardness results for $n$-agent games, and complement them with positive results under the {\it data-driven algorithm design} lens. 

{\bf Central agent improving equilibrium performance}. {Another related line of work considers steering strategic agents to ``good'' equilibria in  a variety of settings, typically with the help of a central agent. \cite{balcan2009improved} consider using a public service advertising campaign, where the central agent prescribes actions to agents and a fraction of the agents (influenced by the campaign) follow actions that could lead to better equilibria. Another variant of the problem is to lead learning dynamics in a certain games where the game happens in phases for the same set of agents, and the impact of the campaign is only assumed in early rounds \cite{balcan2011leading,balcan2013circumventing}. Similarly, \cite{zhang2023steering} consider leading dynamics of agents with vanishing average regret. In contrast, our repeated game settings in Section \ref{sec:data-driven} consist of non-identical but similar games, for example corresponding to different infrastructure projects managed by the same central agent. \cite{kempe2020inducing} consider central agents that can modify the network structure, and study computational tractability for different utility functions and notions of ``good'' equilibrium.}


{\bf Data-driven algorithm design} is a tool for beyond worst-case analysis of algorithms, for learning algorithms that work well on typical problems coming from a common problem domain \cite{gupta2016pac,balcan2017learning,balcan2018dispersion,balcan2020data,sharma2024data}. The technique has been successfully used in designing more effective algorithms for combinatorial problems, with applications to machine learning as well as mechanism design \cite{balcan2018data,balcan2021data,balcan2022provably,balcan2023analysis,balcan2023new,balcan2024learning,morgenstern2015pseudo,balcan2018general}.\purple{The technique has been successfully used in designing more effective algorithms for combinatorial problems \cite{balcan2018learning,lavastida2021learnable}, with applications to machine learning \cite{balcan2018data,balcan2021data,balcan2023analysis} as well as mechanism design \cite{morgenstern2015pseudo,balcan2018general}.} We will use the data-driven algorithm design approach  to learn subsidy schemes from multiple related games. We provide sample complexity guarantees when the games are drawn independently from a fixed distribution, and no-regret guarantees when learning subsidy in an online sequence of games.

\section{Formal notation, setup and motivating examples}\label{sec:prelims}
Let $G=\langle N,(S_i),(\text{cost}_i)\rangle$ denote a game, where $N$ is a set of $n$ agents (or players), $S_i$ is the finite action space of agent $i\in N$, and $\text{cost}_i$ is the cost function of agent $i$. The joint action space of the agents is $S=S_1\times \dots\times S_n$. Given joint action $s=(s_1,\dots,s_n)\in S$ let $s_{-i}$ denote the actions of all agents except agent $i$, i.e. $s_{-i}= (s_1,\dots,s_{i-1},s_{i+1},\dots,s_n)$. The cost function $\text{cost}_i:S\rightarrow \R$ of agent $i$ (which the agent seeks to minimize) is a function of the joint action $s\in S$. The {\it social cost} function of the game is the sum of cost functions of all the agents in the game, $\text{cost}=\sum_{i=1}^n\text{cost}_i$. The optimal social cost is $\OPT = \min_{s\in S} \text{cost}(s)$. 
Given a joint action $s$, the best response of agent $i$ is the set of actions $\BR_{i}(s_{-i})$ that minimizes its cost given $s_{-i}$, i.e., $\BR_i(s_{-i}) = \argmin_{a\in S_i} \text{cost}_i(a,s_{-i})$. A joint action $s\in S$ is a {\it (pure) Nash equilibrium} (or NE) if no agent can benefit from unilaterally deviating to another action, in other words every agent is simultaneously playing a best response action in $s$, i.e., $s_i \in \BR_i(s_{-i})$ for every $i\in N$. A Nash equilibrium  is said to be {\it global } or {\it optimal} if it also minimizes the {\it social cost} among all Nash equilibria. We say a Nash equilibrium is a {\it local} or suboptimal equilibrium if it is not global.

To the above standard model of a game, we add a central agent whose goal is to improve social outcomes by allocating subsidy which reduces costs for certain actions. Formally, we have the following definition of a subsidy scheme.

\begin{definition}[Subsidy scheme]
A subsidy scheme $\bbS$ is defined as a set of functions $\text{subs}_i:S\rightarrow\R_{\ge0}$ where $\text{subs}_i(s)$ gives the subsidy offered by the central agent to agent $i$ given joint action $s$. In a subsidized game using scheme $\bbS$, the cost of agent $i$ is given by the difference function $\text{cost}_i^{\bbS}:=\text{cost}_i-\text{subs}_i$, and total subsidy provided for joint action $s\in S$ is $\text{subs}(s)=\sum_i\text{subs}_i(s)$. 
\end{definition}

The Price of Anarchy measures the reduction in system efficiency (social cost) due to selfish behavior of the agents 
\cite{nisan2007algorithmic,roughgarden2012price}. We define Price of Anarchy (PoA) in the presence of subsidy along the lines of \cite{buchbinder2008non} as the ratio of the sum of total social cost and subsidy in the worst case equilibrium, to the optimal social cost.

\begin{definition}[Price of Anarchy under subsidy]
Let $\mathbb{S}=\{\text{subs}_i\}$ denote the subsidy scheme. Let $\cS_{\text{NE}}(\mathbb{S})\subseteq S$ denote the subset of states corresponding to Nash equilibria when the cost for agent $i$ is $\text{cost}_i-\text{subs}_i$. Suppose $\OPT\ne 0$ and $\cS_{\text{NE}}(\mathbb{S})\ne\{\}$. Then the {\it Price of Anarchy} under subsidy $\bbS$ is given by 

$$\text{PoA}(\mathbb{S})=\frac{\max_{s\in\cS_{NE}(\mathbb{S})}\text{cost}^{\bbS}(s)+\text{subs}(s)}{\OPT}.$$

\noindent We also define a related metric for studying the effectiveness of subsidy scheme $\bbS$, 

$$\widetilde{\text{PoA}}(\mathbb{S})=\frac{\max_{s\in\cS_{NE}(\mathbb{S})}\text{cost}^{\bbS}(s)+\text{subs}(s)}{\min_{s\in\cS_{NE}}\text{cost}(s)},$$

\noindent where $\cS_{NE}$ denotes the set of Nash equilibria in the component maintenance game (in the absence of any subsidy), provided $\min_{s\in\cS_{NE}}\text{cost}(s)\ne 0$ and $\cS_{\text{NE}}(\mathbb{S})\ne\{\}\ne \cS_{NE}$.
\end{definition}

\noindent By setting zero subsidies (i.e. $\text{subs}_i(s)=0$ for each $i,s$) we recover the usual Price of Anarchy, $\text{PoA}$ \cite{nisan2007algorithmic}. Note that finding the subsidy scheme that optimizes $\text{PoA}(\mathbb{S})$ or $\widetilde{\text{PoA}}(\mathbb{S})$ corresponds to the same optimization problem. In some games, it will be easier to show absolute bounds on $\widetilde{\text{PoA}}(\mathbb{S})$. Note that $\widetilde{\text{PoA}}(\mathbb{S})=\text{PoA}(\mathbb{S})/\text{PoS}$, where $\text{PoS}$ is the usual Price of Stability in the unsubsidized game \cite{nisan2007algorithmic}.  

Besides PoA, we will also be interested in another quantity called {\it Value of Information} \cite{lin2021multi} which we define next. Suppose the costs of agents have some uncertainty, which we model by a prior  $\theta$ (common belief shared by all the agents) over some finite information set $\cI$ about the game (e.g.\ the component states in a component maintenance game, or  action costs in a cost-sharing game, see below). The expected cost of agent $j$ under joint action $s$ is given by $l_j(s,\theta)=\bbE_{I\sim \theta}[\text{cost}_j(s)]$. We will refer to this as the {\it prior game}. The Value of Information for agent $j$ when information $i\in\cI$ is revealed is the change in expected cost of the agent from the prior $\theta$ to the posterior $\theta^{i,e_i}$, where $e_i$ denotes the revealed value of the information $i$. For example, in the component maintenance  game defined below the information $i$ corresponds to the some agent's component and $e_i$ corresponds to the revealed state (working or broken). The game with expected costs given by $l_j(s,\theta^{i,e_i})$ is called the {\it posterior game}. Formally, the value of information is defined as follows.

\begin{definition}[Value of Information] Denote by $\theta$ the prior, and by $\theta^{i,e_i}$ the posterior when information $i\in\cI$ is revealed to be $e_i$. Let $s,\tilde{s}$ be joint actions which are Nash Equilibria in the prior and posterior games respectively. The Value of Information for agent $j$ when information $i$ is revealed as  $e_i$ is given by 
    $\text{VoI}_{j,i}(s,\tilde{s}) := l_j(s,\theta)-l_j(\Tilde{s},\theta^{i,e_i})$. We say that the Value of Information is non-negative for agent $j$ if $\text{VoI}_{j,i}(s,\tilde{s})\ge 0$ for  any information $i\in\cI$ and any prior/posterior equilibria $s,\tilde{s}$. The worst-case Value of Information is defined as $\min_{s,\tilde{s}} \text{VoI}_{j,i}(s,\tilde{s})$, where the minimum is over joint actions from prior and posterior Nash Equilibria. We will often call this simply the Value of Information, the worst-case aspect will be clear from context (lack of explicit arguments $s,\tilde{s}$).
\end{definition}

We will be interested in a collection of related game instances, specifically the sample complexity of number of game instances needed from a distribution over the games to learn a good value of subsidy. Formal definitions follow the standard in data-driven algorithm design, and are deferred to Section \ref{sec:sample-complexity-defs}. We will now proceed to formally describe the component maintenance and cost-sharing games and instantiate the above abstract definitions for both.

{\it Component maintenance game} \cite{lin2021multi}. 
Each agent is associated with a component $c_i$ which has a binary state $x_i\in\{0,1\}$ where $x_i=0$ corresponds to a broken component and $x_i=1$ corresponds to a functioning component. The action space of each agent is also binary, $S_i=\{0,1\}$, where action $s_i=1$ indicates that the agent repaired the component (denoted RE), and $s_i=0$ denotes that the agent did nothing (denoted DN). The state $x_i$ of $c_i$ is updated after action $s_i$ as $x'_i=\max\{x_i,s_i\}$. This corresponds to ``perfect repair'', i.e.\ if an agent picks the RE action, their component is guaranteed to work, and otherwise it stays as is. For a tuple of actions $s$, we will denote the updated state by $\x'(s)$, or simply $\x'$ when $s$ is evident from context. The state $u$ of the system is a fixed binary function of the component states, $u=\phi(\x)$, where $\x=(x_1,\dots,x_n)$ and $\phi:\{0,1\}^n\rightarrow\{0,1\}$\footnote{The component maintenance game intuitively corresponds to monotone boolean functions $\phi$, but our results  easily extend to general boolean functions.}. For example, if $\phi(x_1,x_2,x_3)=(x_1\land x_2)\lor x_3$, the system functions either when both components $c_1$ and $c_2$ are working, or when component $c_3$ is working. Here $u=0$ denotes a failure of the system. Let $u'=\phi(\x')$ denote the state of the system after the agents' actions. The cost for agent $i$ is given by, $\text{cost}_i=C_is_i+1-u'$ for repair cost $C_i\in\R$. {Note that $C_i$ could be negative, for example if there is a reward or incentive associated with the repair of component $i$ which more than offsets what the agent pays for its repair.} The actions depend on the {\it belief} about the state $\x$ of the components, which we model by a distribution $\theta$ over $\{0,1\}^n$. We will assume here that components and therefore their probability of functioning are independent, and that all agents share the same common belief about the state of the components, i.e.\ $\theta$ is fixed and known to all the agents. The system failure probability under this belief is $P_{\phi}(\theta)=1-\bbE_{\x\sim \theta}[\phi(\x)]$, and the expected cost of action $s_i$ to agent $i$, given other agents' actions are $s_{-i}$, is $l_i(s_i,s_{-i},\theta)=\bbE_{\x\sim \theta}[\text{cost}_i]=C_is_i+1-\bbE_{\x\sim \theta}[\phi(\x')]$. Similarly, the expected social cost is defined as $l(s,\theta)=\bbE_{\x\sim \theta}[\text{cost}]=\sum_il_i(s_i,s_{-i},\theta)$ for $s=(s_1,\dots,s_n)$. 

We now imagine that each agent $j$ has the ability, for free, to inspect their own component to determine its state.  The catch, however, is that this state is revealed to all agents. The revelation of the state would result in an updated common belief $\tilde{\theta}$, corresponding to the conditional distribution given the known component state, and therefore an updated cost function for all agents. As a result, the set of Nash equilibria can now change and the agent may suffer higher personal cost in the new equilibria, inducing the agent to avoid inspecting their own component for selfish reasons. We will use the terminology {\it prior} (respectively {\it posterior}) game and equilibria to refer to the game before (respectively after) the inspection of any fixed component $c_j$. We next formalize the setup and instantiate the Value of Information metric for component maintenance game  to capture this behavior.

{\it Component inspection game} \cite{lin2021multi}. Suppose the inspection of component $c_j$ reveals state $y_j\in\{0,1\}$. We will assume perfect inspection, i.e.\ $y_j=x_j$. Denote posterior belief after the inspection of component $c_j$ by $\Tilde{\theta}^{j,y_j}$. If agent $i$ switches action from $s_i$ to $\Tilde{s}_i^{j,y_j}$ after the inspection (and other agents switch from $s_{-i}$ to $\Tilde{s}_{-i}^{j,y_j}$), the value of information about inspection of $c_j$ for agent $i$ and posterior $y_j$ is given by $\text{VoI}_{i,j}(s_i,s_{-i},\Tilde{s}_i^{j,y_j},\Tilde{s}^{j,y_j}_{-i}) := l_i(s_i,s_{-i},\theta)-l_i(\Tilde{s}_i^{j,y_j},\Tilde{s}_{-i}^{j,y_j},\Tilde{\theta}^{j,y_j})$. The expected value of information is given by $\overline{\text{VoI}}_{i,j}(s_i,s_{-i},\Tilde{s}_i^{j,*},\Tilde{s}_{-i}^{j,*}) := l_i(s_i,s_{-i},\theta)-\bbE_{y_j}l_i(\Tilde{s}_i^{j,y_j},\Tilde{s}_{-i}^{j,y_j},\Tilde{\theta}^{j,y_j})$, where $\Tilde{s}^{j,*}$ is the collection of states $\Tilde{s}^{j,0},\Tilde{s}^{j,1}$. Here the posterior loss is computed as the expectation (based on prior) for the various possible outcomes for the inspection of a given component $j$. Typically we will assume that the joint actions $s$ and $\Tilde{s}^{j,y_j}$ are Nash equilibria. We want the value of information to be non-negative for each agent $i$, when inspecting any component $j$, for any choice of equilibrium states. The expected value of information is easier to ensure to be non-negative, but implies a weaker (less robust) guarantee, so we will focus on the (inspection-specific) value of information. A motivation for ensuring that the value of information is non-negative is to not have any undesirable information avoidance behavior among the agents, where agents may choose to ignore freely available information (about the inspected component state) for selfish reasons (to reduce personal cost, for example by choosing to not repair their broken component), which could lead to sub-optimal social cost or an undesirable system state.

\begin{figure}[t]
  \centering
  \includegraphics[width=0.35\linewidth]{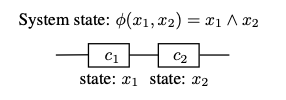}
  \caption{A two component series system.\label{fig: series}}
\end{figure}

{\it Fair-cost sharing game}. Let $\cA=\{a_j\}$ denote the (finite) set of all possible actions for all the agents. There is a function $f:\cA\rightarrow 2^N$ such that 
agent $i\in[N]$ may use any action $a_j$ for which $i\in f(a_j)$. Note that there may be multiple options corresponding to the same subset of agents. Under {\it uniform} or fair cost-sharing, all agents that use an action $a_j$ in some state $s\in S$ equally share its cost. That is, in the classical fair cost-sharing game there is a deterministic function $c:\cA\rightarrow[0,C_{\max}]$ such that if $k$ agents from $f(a_j)$ use an action $a_j$ in some state $s$, then $\text{cost}_i(s)$ for each of these agents is $\frac{c(a_j)}{k}$. Here, we will consider a Bayesian extension where the costs of some actions in $\cA$ are associated with some uncertainty. The action costs are given by a distribution $\theta_c$ over $[0,C_{\max}]^{|\cA|}$ and the agents all know the costs under the prior given by $l_i(s,\theta_c)=\frac{1}{k}\bbE_{\theta_c}[c(a_j)]$, where $k$ is the number of agents opting for action $a_j$ in state $s$. This means we expect the agents to act according to the mean costs of actions under $\theta_c$ in the absence of additional information.

Analogous to the component inspection game above, suppose the inspection of action $a_j$ reveals its true cost $c_j\in[0,C_{\max}]$. Denote posterior belief after the revelation of the cost of action $a_j$ by $\Tilde{\theta}_c^{j,c_j}$. If agent $i$ switches action from $s_i$ to $\Tilde{s}_i^{j,c_j}$ after the inspection (and other agents switch from $s_{-i}$ to $\Tilde{s}_{-i}^{j,c_j}$), the value of information about inspection of $a_j$ for agent $i$ and posterior $c_j$ is given by $\text{VoI}_{i,j}(s_i,s_{-i},\Tilde{s}_i^{j,c_j},\Tilde{s}^{j,c_j}_{-i}) := l_i(s_i,s_{-i},\theta_c)-l_i(\Tilde{s}_i^{j,c_j},\Tilde{s}_{-i}^{j,c_j},\Tilde{\theta_c}^{j,c_j})$. The expected value of information is given by $\overline{\text{VoI}}_{i,j}(s_i,s_{-i},\Tilde{s}_i^{j,*},\Tilde{s}_{-i}^{j,*}) := l_i(s_i,s_{-i},\theta_c)-\bbE_{c_j}l_i(\Tilde{s}_i^{j,c_j},\Tilde{s}_{-i}^{j,c_j},\Tilde{\theta_c}^{j,c_j})$, where $\Tilde{s}^{j,*}$ is the collection of posterior states $\Tilde{s}^{j,c_j}$ for different posterior costs $c_j$. As before, typically we will assume that the joint actions $s$ and $\Tilde{s}^{j,c_j}$ are Nash equilibria. 

\subsection{Motivating examples for how subsidy can help}\label{sec:subsidy}

{\bf Component maintenance game}. We first present a motivating example where subsidy for component repair costs can help improve social cost by steering the system to a better equilibrium.
\blue{
\cite{lin2021multi} give examples of local equilibria for two agent games where the Value of Information (VoI is defined as difference of prior loss defined above, and posterior loss when a component is inspected) is negative. This is undesirable as it can trigger the agents to avoid the perfect information on component state for selfish reasons. We present a motivating example where subsidy for component repair costs can help improve social cost as well as prevent information avoiding behavior on inspection.}

\begin{example}[PoA in a 2-series system] Consider the two component series system depicted in Fig \ref{fig: series}. Suppose that the components are independent and
the failure probabilities (according to the common prior belief) for components $c_1$ and $c_2$ are both 0.5, and the repair cost is
$C_1 = C_2 = 0.3$. Then the cost matrix for the component maintenance game is given in Table \ref{tab:2-examples}. For instance, for the joint action DN-DN (recall that action DN stands for ``do nothing'' and RE for ``repair'') the system works with probability $0.5\times 0.5=0.25$ and therefore the cost to agents is $P_\phi(\theta)=1-0.25=0.75$. \purple{explain table} Notice that both DN-DN and RE-RE are Nash equilibria for the  game, but RE-RE has a smaller social cost.  If the central authority provides a subsidy of $0.05+\epsilon$ for the repair action, for any $\epsilon>0$, then it would incentivize the agents to switch their actions from DN to RE, and the only Nash equilibrium is RE-RE. Note that the cost reduction in RE-RE (0.3 to $0.25-\epsilon$ for each agent) equals the subsidy provided by the central agent in this case, and the social cost + subsidy for RE-RE is preserved, while ruling out the sub-optimal equilibrium DN-DN.
\label{ex:2-series}
\end{example}


\begin{table}[t]
\centering
\caption{Cost matrix for (agent $1$, agent $2$) for the two-agent game with components connected in series.}
\begin{tabular}{@{}lcc@{}}
\toprule
Agent 1\textbackslash 2 & DN   ($s_2=0$)         $\;$         & RE   ($s_2=1$)   \\ \midrule
DN   ($s_1=0$) $\;$   & $0.75,0.75$ & $0.5,0.8$     \\
RE  $(s_1=1$) $\;$  & $0.8,0.5$       & $0.3,0.3$       \\
 \bottomrule
\end{tabular}
\label{tab:2-examples}
\end{table}









\noindent The above example illustrates how using a subsidy scheme, the central agent can potentially eliminate undesirable Nash equilibria by effectively modifying the repair costs in the component maintenance game. 
 In the component maintenance game, we will consider subsidy schemes that incentivize repair, i.e.\  $\text{subs}_i(0,s_{-i})=0$ for all agents $i$ (no subsidy awarded for the action ``do nothing''). We also say that a subsidy scheme is {\it uniform} if the scheme is identical for all agents and actions, i.e. $\text{subs}_i(1,s_{-i})=c_{\text{subs}}$ a fixed constant for all agents $i$ for some constant $c_{\text{subs}}\ge 0$. We consider two types of subsidies which a central agent, whose goal is to altruistically maximize social welfare, can offer to reduce the repair cost of certain components.

{\it Conditional vs.\ unconditional subsidies.}  The central agent may offer an {\it unconditional} subsidy which effectively reduces the cost of repair for the components, or may be {\it conditional on inspection} in order to encourage agents to inspect their components, even when the information about the state of an agent's component results is something the agent might want to avoid (in the absence of subsidy). 
Formally, in a  component inspection game involving inspection of some component $j$, a general  subsidy consists of three functions for each agent $i$ given by $\text{subs}_i$, $\text{subs}_i^1$, $\text{subs}_i^0$, corresponding to prior, posterior with component $j$ intact, and posterior with component $j$ damaged respectively. For simplicity, we will say that the central agent provides subsidy conditional on inspected state of component $j$, $y_j=k$ for $k\in\{0,1\}$, to denote that $\text{subs}_i^k$ is the only non-zero function in the conditional scheme, and conditional on inspection to denote that $\text{subs}_i$ is a zero function and $\text{subs}_i^1=\text{subs}_i^0$.

We will now illustrate how subsidy can be used to avoid negative Value of Information, again in a 2-series component inspection game.

\begin{example}
    Consider the two component series system from \cite{lin2021multi}. Suppose that the
components are independent and the failure probabilities for $c_1$ and $c_2$ are 0.6 and 0.9
respectively, and the repair cost is $C_1 = C_2 = 0.3$. Then the cost matrix for the game
(both prior, and posterior cost when component $c_1$ is inspected) is given in Table \ref{tab: voi-example-cig}. Without any subsidy, for the highlighted Nash equilibria, when component $c_1$ is inspected and revealed to be broken ($y_1=0$), the value of information is negative ($0.3-1.0=-0.7$) for both agents.
If the central authority gives agent 1 an incentive of $0.2 + \epsilon$ for any $\epsilon > 0$, then it
would incentivize agent 1 to switch their action from DN to RE in the posterior game
$y_1 = 0$ (which will cost $1.0 - \epsilon$ after the subsidy) and negative VoI can be avoided in
this example. The expected VoI is $0.3-(0.6*1.0+0.4*0.0)=-0.3$ for agent 1 and $0.3-(0.6*1.0+0.4*0.3)=-0.42$ for agent 2, and the same subsidy works to ensure the expected VoI is non-negative as well.
\end{example}

\begin{table}[t]
\centering
\caption{Cost matrix for (agent $1$, agent $2$) for the two-agent component inspection game with components connected in series. Selected Nash Equilibria used in calcuation of Value of Information are highlighted in red.}
\begin{tabular}{@{}lcccc@{}}
\toprule
Condition/Joint action & DN-DN & DN-RE & RE-DN        & RE-RE      \\ \midrule
Prior   & $0.96,0.96$ & $0.6,0.9$ & $1.2,0.9$ & \red{$0.3,0.3$}     \\
Posterior  $(y_1=1$)   & $0.90,0.90$ & \red{$0.0,0.3$} & $1.2,0.9$ & $0.3,0.3$     \\
Posterior  $(y_1=0$)   & \red{$1.0,1.0$} & $1,1.3$ & $1.2,0.9$ & $0.3,0.3$     \\
 \bottomrule
\end{tabular}
\label{tab: voi-example-cig}
\end{table}

{\bf Cost-sharing game}. In the {\it cost-sharing game},  the central agent can subsidize the cost of some actions in $\cA$. That is, the subsidy can be specified as an allocation of the subsidy budget to actions, $c^{\bbS}:\cA\rightarrow\R_{\ge 0}$. The corresponding {\it subsidy scheme} is given by $\{\text{subs}_i\}_i$, where $\text{subs}_i(s_1,\dots,s_N)=\frac{c^\bbS(s_i)}{k}$, with $k=\sum_{j=1}^N\I[s_j=s_i]$. This corresponds to the subsidy being equally enjoyed by all the agents that select a given subsidized action. Subsidy can again be used to reduce Price of Anarchy in this game \cite{buchbinder2008non}. We will here show an example where subsidy can be used to avoid negative Value of Information.

\begin{example}
Consider a two-agent cost-sharing game where the action set is $\cA=\{A,B,C,D\}$ with associated subsets $f(A)=\{1,2\}$, $f(B)=\{2\}$, $f(C)=\{2\}$, and $f(D)=\{1\}$. For example, in a commuting game, $A$ could correspond to a shared public transport, and $B,C,D$ could correspond to private modes of transport. We assume the cost function $c$ is a random function such that with probability $\frac{1}{2}$, $c(A)=5$, $c(B)=2$, $c(C)=6$, and $c(D)=4$, and with probability $\frac{1}{2}$, $c(A)=5$, $c(B)=6$, $c(C)=2$, and $c(D)=4$. In the commute example, for agent 2, $B$ could be a bike and $C$ could be a car, and $w_i$ could be unknown world state that impacts the cost of actions $B$ and $C$ for agent 2. We call these posterior worlds $w_1$ and $w_2$ respectively (see Tables \ref{tab:2-cs1} and \ref{tab:2-cs2}). Thus the prior cost for agent 2 (i.e. probability weighted cost of the worlds $w_1$ and $w_2$) is 4 for actions $B$ and $C$. A Nash equilibrium in the prior game is $(A,A)$ with a cost of $(2.5,2.5)$ for both agents. In world $w_1$, the only NE is $(D,B)$ and in world $w_2$ the only NE is $(D,C)$, both with cost $(4,2)$. Thus the knowledge of the state of the world leads to negative VoI for agent 1. Specifically, the knowledge of the cheaper option among $B$ and $C$ causes agent 2 to drop out of the cost-shared option $A$, increasing social cost and cost for agent 1.     In this example, using a subsidy of $3+\epsilon$ for $\epsilon>0$ for the option $A$ guarantees that agent $2$ will always prefer option $A$, and is sufficient to ensure that negative Value of Information 
is avoided.
\end{example}

\begin{table}[]
\centering

\caption{Action-cost matrix for the two-agent cost sharing game in Example 3, for the prior game, posterior $w_1$ and posterior $w_2$. Selected NE is highlighted in red.\label{tab:2-prior}}
\begin{tabular}{@{}cccc@{}}
\toprule
Agent 1 \textbackslash\ Agent 2 & $A$                     & $B$                   & $C$                                 \\ \midrule
$A$      & \red{$2.5,2.5$} & $5,4$       & $5,4$  \\
$D$    & $4,5$       & $4,4$       & $4,4$ \\ \bottomrule
\end{tabular}

\vspace*{1em}
\begin{tabular}{@{}cccc@{}}
\toprule
Agent 1 \textbackslash\ Agent 2 & $A$                     & $B$                   & $C$                                 \\ \midrule
$A$      & $2.5,2.5$ & $5,2$       & $5,6$  \\
$D$    & $4,5$       & \red{$4,2$}       & $4,6$ \\ \bottomrule
\end{tabular}

\vspace*{1em}

\begin{tabular}{@{}cccc@{}}
\toprule
Agent 1 \textbackslash\ Agent 2 & $A$                     & $B$                   & $C$                           \\ \midrule
$A$      & $2.5,2.5$ & $5,6$       & $5,2$  \\
$D$    & $4,5$       & $4,6$       & \red{$4,2$} \\ \bottomrule
\end{tabular}
\end{table}

We conclude this section with a simple remark that $\text{PoA}(\mathbb{S})\ge 1$ for any subsidy scheme $\mathbb{S}$ since the amount of subsidized cost is added back in the numerator for any state $s$ in the definition of $\text{PoA}(\mathbb{S})$, and the dependence on the subsidy scheme $\mathbb{S}$ is governed by the corresponding set of Nash Equilibria $\cS_{NE}(\mathbb{S})$.

\begin{proposition}
    For any subsidy scheme $\mathbb{S}$, $1\le \text{PoA}(\mathbb{S})$ and $\widetilde{\text{PoA}}(\mathbb{S})\le \text{PoA}(\mathbb{S})$.
\end{proposition}
\begin{proof}
    Let $\mathbb{S}=\{\text{subs}_i\}$ denote the subsidy scheme. By definition,
    \begin{align*}
        \text{PoA}(\mathbb{S})&=\frac{\max_{s\in\cS_{NE}(\mathbb{S})}\text{cost}^{\bbS}(s)+\text{subs}(s)}{\OPT}\\
        &=\frac{\max_{s\in\cS_{NE}(\mathbb{S})}\sum_i(\text{cost}_i(s)-\text{subs}_i(s))+\text{subs}_i(s)}{\OPT}\\
        &=\frac{\max_{s\in\cS_{NE}(\mathbb{S})}\text{cost}(s)}{\OPT}\ge \frac{\min_{s\in\cS_{NE}(\mathbb{S})}\text{cost}(s)}{\OPT}\ge 1.
    \end{align*}
\noindent Also, $\text{PoA}(\mathbb{S})\ge \widetilde{\text{PoA}}(\mathbb{S})$ since ${\min_{s\in\cS_{NE}}\text{cost}(s)}\ge {\min_{s\in S}\text{cost}(s)}=\OPT$.
\end{proof}

Addition of subsidy effectively changes the set of states that correspond to a Nash equilibrium. The goal of a central agent is to ensure suboptimal (local) Nash equilibria are not included in the set $\cS_{NE}$ when the subsidy is applied. The subsidy provided by the central agent could for example come from taxes collected from the agents, and therefore it makes sense to add the   total subsidy provided by the central agent to the social cost in the definition of $\text{PoA}(\mathbb{S})$. 

For a fixed game, the appropriate subsidy scheme could vary depending on the objective of the central agent. We will illustrate this for a two agent series game in Section \ref{sec:opt-subs-2-series}, where we will  obtain the optimal subsidy schemes for different objectives below. Similarly, optimal schemes can be derived for the two-agent parallel game, the interested reader is directed to Appendix \ref{app:2-parallel}. However, for  general $n$-agent games, we  show in Section \ref{sec:hardness} that computing the optimal total subsidy is NP hard, under \purple{each of }the various objectives.

\subsection{Optimal subsidy design in two-agent series component maintenance game}\label{sec:opt-subs-2-series}
This section aims to provide a working intuition for what an optimal subsidy scheme looks like for simple two-agent component maintenance games for the different objectives of the central agent, including Price of Anarchy and Value of Information. The key insights are that subsidy can be greatly beneficial, but exactly optimal subsidy allocation is a complicated function of the game's cost matrix and objective of interest. In the two-agent series game, the goal is typically to provide just enough subsidy to ensure that both agents prefer to repair (and the series system functions as a result), but the smallest amount of subsidy needed depends on repair costs and component failure probabilities.

Suppose we have two agents $N=\{1,2\}$ with components $c_1,c_2$ connected in series. Let $p_1,p_2$ denote the (prior) probability that the components $c_1,c_2$ will work (respectively). 
We will use the notation $\overline{p}:=1-p$ for conciseness. 
This corresponds to the prior game in Table \ref{tab:2-series}, where DN denotes ``do nothing'' ($s_i=0$) and RE denotes ``repair'' ($s_i=1$). 
We will now consider the above three objectives. For each objective, we will note that subsidy helps and we will obtain the optimal subsidy in the two-agent series game.\looseness-1

\subsubsection{Minimizing the price of anarchy with subsidy.}
To demonstrate the significance of subsidy for reducing the social cost in the two-agent series games, we will first show a lower bound on the Price of Anarchy in the absence of subsidy. The following proposition indicates that the Price of Anarchy can be very high when the probabilities of the components functioning ($p_1,p_2$) are small. Moreover, for $n$ agents connected in series, the price of anarchy can increase exponentially with number of agents $n$.

\begin{proposition}\label{prop:poa-no-subs} In the two-agent series prior game (defined above and cost matrix noted in the first row of Table \ref{tab:2-series}), the Price of Anarchy in the absence of subsidy is at least $\text{PoA}\ge \frac{2}{p_1+p_2}$, 
for some repair costs $C_1,C_2$. More generally, for n agents, $\text{PoA}\ge \tilde{H}/\tilde{G}^n$ for some repair costs $C_1,\dots,C_n$, where $\tilde{H}$ and $\tilde{G}$ are the harmonic and geometric means, respectively, of the prior probabilities $p_1,\dots,p_n$.
\end{proposition}

\noindent By Proposition \ref{prop:poa-no-subs}, if probabilities of component working $p_1=p_2={\epsilon}\ll 1$, then Price of Anarchy can be as large as $\frac{1}{\epsilon}$ (or as large as $\epsilon^{-(n-1)}$ for general $n$). We remark that our lower bounds also apply to the ratio PoA/PoS, i.e.\ when we compare the worst-case Nash equilibrium with the {\it global} Nash equilibrium instead of $\OPT$. We will now show that a constant price of anarchy can be achieved in this game using subsidy, more precisely, that $\widetilde{\text{PoA}}(\bbS)=1$ for some subsidy scheme $\bbS$. In fact, we characterize the total subsidy needed to guarantee this for any game parameters. The proof involves carefully considering cases for the parameters resulting in different sets of Nash equilibria and computing necessary and sufficient amounts of subsidy in each case (see Appendix \ref{app:2-series}).

\begin{theorem}\label{thm:poa-series}
    Consider the two-agent series component maintenance game with $C_1,C_2>0$ and $0<p_1,p_2<1$. Let $s^*=\I\{(C_1,C_2)\in [\overline{p_1}p_2,\overline{p_1}]\times [\overline{p_2}p_1,\overline{p_2}]\}\cdot\min\{C_1-\overline{p_1}p_2, C_2-\overline{p_2}p_1\}$, where $\I\{\cdot\}$ denotes the 0-1 valued indicator function. Then there exists a subsidy scheme $\bbS$  with total subsidy $s$ for any $s>s^*$ such that $\widetilde{\text{PoA}}(\mathbb{S})=1$. Moreover, a total subsidy of at least $s^*$ is necessary for any subsidy scheme $\bbS$ that guarantees $\widetilde{\text{PoA}}(\mathbb{S})=1$. 
\end{theorem}


\begin{figure}[t]
  \centering
  \includegraphics[width=0.35\linewidth]{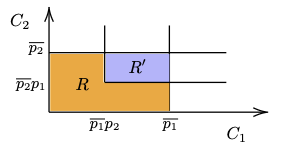}
  \caption{Cost region $R$ with system functioning in any NE ($\phi(\x')=1$, shaded orange) in a two-agent series game (Theorem \ref{thm:opt-series}).  $\text{PoA}\ne\text{PoS}$ in region $R'$ (Theorem \ref{thm:poa-series}). \label{figure: proof}}
\end{figure}

\subsubsection{Guaranteeing that the system functions in any NE} As seen in Example \ref{ex:2-series}, the agents may be in an equilibrium (e.g. DN-DN) such that the system is not guaranteed to function. This problem can also be remedied using subsidy.
We will quantify the optimal subsidy needed to guarantee that the system functions, i.e.\ $\phi(\x')=1$, where $\x'$ denotes the component states after the agents' actions. 

\begin{theorem}\label{thm:opt-series}
    Consider the two-agent series component maintenance game with $C_1,C_2>0$ and $0<p_1,p_2<1$. Define $R\subset\R^+\times \R^+$ as the set of cost pairs that satisfy $C_1\le \overline{p_1}\land C_2\le \overline{p_2} \land (C_1\le \overline{p_1}p_2\lor C_2\le \overline{p_2}p_1)$, depicted in Figure \ref{figure: proof}. Let $s^*=\min_{(x,y)\in R}||(C_1,C_2)-(x,y)||_1$, where $||\cdot||_1$ denotes the $L_1$-norm. Then there exists a subsidy scheme $\bbS$  with total subsidy $s$ for any $s>s^*$ such that the system functions in any NE. Moreover, a total subsidy of at least $s^*$ is necessary for any subsidy scheme $\bbS$  that guarantees that the system functions in any NE.
\end{theorem}

\noindent The proof is deferred to Appendix \ref{app:2-series}. In constrast to Theorem \ref{thm:poa-series}, here the central agent needs to provide subsidy in the states where the price of anarchy may be 1 without subsidy but the system is not guaranteed to function as agents may choose to do nothing.

\subsubsection{Ensuring value of information is non-negative for each agent.} \cite{lin2021multi} exhibit several examples, including 2-agent series games, where the value of information for certain agents can be negative, when actions in the prior and posterior games are selected according to some Nash equilibria. We will demostrate the use of subsidy in tackling this undesirable information avoidance behavior. 

In addition to the prior game, we will now consider posterior games where component $c_1$ is inspected, and its state $y_1$ is revealed on inspection ($y_1=1$ corresponds to $c_1$ is functioning, and $y_1=0$ corresponds to $c_1$ is broken). Table \ref{tab:2-series} summarizes prior and posterior costs for the two agents for each action pair (recall that DN denotes ``do nothing'' or $s_i=0$, and RE denotes ``repair'' or $s_i=1$). \cite{lin2021multi} show that the expected value of information (VoI) is non-negative for all agents if 
a global NE is selected.

\begin{table}[t]
\centering
\caption{Matrix for cost-pairs (agent $1$, agent $2$) when component $c_1$ is inspected for the two-agent series system. Here $\overline{p_1p_2}=1-p_1p_2$ and $\overline{p_i}=1-p_i$.}
\begin{tabular}{@{}lcccc@{}}
\toprule
Conditions & $\quad$ DN-DN   $\quad$                  & $\quad$ DN-RE $\quad$                  & $\quad$ RE-DN   $\quad$                & $\quad$ RE-RE  $\quad$               \\ \midrule
Prior      & $\overline{p_1p_2},\overline{p_1p_2}$ & $\overline{p_1},\overline{p_1}+C_2$ & $\overline{p_2}+C_1,\overline{p_2}$ & $C_1,C_2$ \\
$y_1=1$    & $\overline{p_2},\overline{p_2}$       & $0,C_2$       & $\overline{p_2}+C_1,\overline{p_2}$ & $C_1,C_2$ \\
$y_1=0$    & $1,1$                     & $1,1+C_2$      & $\overline{p_2}+C_1,\overline{p_2}$               & $C_1,C_2$ \\ \bottomrule
\end{tabular}
\vspace*{2mm}
\label{tab:2-series}
\end{table}


\begin{theorem}[\cite{lin2021multi}]\label{thm:global-voi}
In the two-agent series game described above, if component $c_j,j\in\{1,2\}$ is inspected, and the prior and posterior actions $s,\tilde{s}$ are selected from global  equilibria, then the expected value of information VoI$_{i,j}(s,\tilde{s})$ is non-negative for each agent $i\in\{1,2\}$. 
\end{theorem}

\noindent Therefore, if we avoid suboptimal local equilibria in this setting then expected VoI is guaranteed to be non-negative. Combined with Theorem \ref{thm:poa-series} above, this implies that negative Value of Information may be avoided via subsidizing repair costs in the two-agent series game. In more detail, we can employ a conditional subsidy scheme, with prior subsidy $\text{subs}_i$ according to the scheme $\bbS$ in Theorem \ref{thm:poa-series}, and posterior subsidy $\text{subs}_i^y$ when agent 1 is inspected employs the subsidy scheme from Theorem \ref{thm:poa-series} by setting the parameter $p_1=y$. Theorem \ref{thm:voi-series} in Appendix \ref{app:2-series} characterizes the optimal unconditional subsidy for ensuring the Value of Information is non-negative for each agent.

\section{Computational Hardness of Subsidy Design}\label{sec:hardness}

We will now consider the problem of computing the best subsidy scheme  in  general component maintenance  and cost-sharing games. 
We will prove  computational hardness results in this section and assume familiarity with fundamental  concepts in complexity theory \cite{arora2009computational}. We will do this by reducing the \textsc{Vertex-Cover} problem (one of Karp's original NP-complete problems \cite{karp2010reducibility}) to the decision-problem version of computing the best subsidy for the component maintenance game and from the \textsc{Min-Set-Cover} problem for the cost-sharing game. 
Recall that  \textsc{Vertex-Cover} is the following decision problem, specified by a graph $\cG=(V,E)$ and an integer $k$.

\textsc{Vertex-Cover}: Does a given (undirected, unweighted) graph $\cG=(V,E)$ admit a vertex cover\footnote{A set of vertices such that all edges of the the graph have at least one endpoint in the set.} of size $k$?

\noindent We also consider the optimization version of the minimum subset cover problem $(\cU,\cS,k)$ stated below, which will be useful in our hardness results for the cost-sharing game.

\textsc{Min-Set-Cover}: Given a finite set $\cU$ of size $n$ and a collection $\cS\subseteq2^\cU$ of subsets of $\cU$, does there exist a subset $S$ of $\cS$ of size $k<n$ that covers $\cU$, i.e.\ $\cup_{S_i\in S}S_i=\cU$?

In the next three subsections we will show computational hardness results for designing the optimal subsidy scheme for several different objectives of interest to the central agent,  including the Price of Anarchy under subsidy, guaranteeing the Value of Information is non-negative for all agents, and (in the component maintenance game) ensuring the system functions in any Nash equilibrium.

\subsection{Price of Anarchy under Subsidy} 

{\bf Component maintanence game}. Finding the best  subsidy in a given component maintanence game (CMG) to minimize the price of anarchy is an optimization problem. Here we will study the hardness of a corresponding decision problem stated below.

\textsc{CMG-PoAS}: Given a CMG $G$ and a subsidy budget  $n^*$, 
does there exist  subsidy scheme $\bbS$ with non-zero subsidy provided to  $n^*$ agents such that $\text{PoA}(\bbS)=1$?

Note that here we define the best subsidy scheme to be the one that provides subsidy to fewest agents (as opposed to providing the smallest total subsidy paid out by the central agent, which we also study subsequently). This meaningfully models situations when subsidy consists of allocation of indivisible resources, for example the central agent providing commitment to intervene and assist but with a bandwidth constraint on the number of agents that could receive the assistance. The question whether the same hardness result can be established when the subsidy budget is continuous is left open in this case (we do consider real-valued subsidy budgets for other objectives, see Sections \ref{sec:hard-voi}, \ref{sec:hard-sys}). As mentioned above, we will do a Karp reduction from the \textsc{Vertex-Cover} problem. 
The key idea is to construct a component maintenance game $G$ given any graph $\cG$, with agents corresponding to graph nodes, and the system state function $\phi$ corresponding to the graph edges. We show that the Nash equilibria of $G$ roughly correspond to minimal vertex covers of $\cG$. To ensure $\text{PoA}(\bbS)=1$, we provide subsidy to the agents in a smallest vertex cover of $\cG$.

\begin{theorem}\label{thm:cmg-poa-hard}
    \textsc{CMG-PoAS} is NP-Hard.
\end{theorem}
\begin{proof}
We will reduce the \textsc{Vertex-Cover} problem to \textsc{CMG-PoAS}.  Given an instance $\cG,k$ of the \textsc{Vertex-Cover} problem, we create a corresponding \textsc{CMG-PoAS} problem as follows. Introduce an agent $i$ for every vertex  $i\in V$ and consider the (2-CNF) formula ${\phi(\x)}=\bigwedge_{(i,j)\in E}(x_i\lor x_j)$, where the clauses consist of component states $x_i,x_j$ for all pairs $i,j$ of agents corresponding to edges in $E$. Set the probability distribution $\theta$ to be the constant distribution with the entire probability mass on $0^n$ (i.e.\ all the components are guaranteed to fail without repair). Set repair cost $C_i=1$ for all components $i$. Then the cost function for agent $i$ for joint action $s=(s_i,s_{-i})$ is given by\looseness-1
\begin{align*}
    l_i(s_i,s_{-i},\theta)=\bbE_{\x\sim \theta}[\text{cost}_i]&=C_is_i+P_{\phi}(\theta)\\&=1\cdot s_i+1-\phi(\x')\\&=s_i+1-\phi(s),
\end{align*} where $x_i'=\max\{0,s_i\}=s_i$ denotes the  component state after agent $i$ takes action $s_i$.

We first proceed to characterize the set of Nash equilibria of this game. Note that  $s=0^n$ is a NE for this game, since any repair action by any agent increases the agent's cost by $1$ if the repair does not change the state $\phi$ of the system, and by $0$ otherwise. This is because $\phi$ is monotonic, $\phi(s)$ can only change from 0 to 1 on repair. For the same reason, no agent has any incentive to switch from DN to RE. We will now show that the remaining NEs for the game correspond to minimal vertex covers of $\cG$. 

Suppose $K\subseteq[n]$ is a set of agents for which the corresponding nodes in $\cG$ constitute a minimal vertex cover. Let $s_{K'}=(s_1,\dots,s_n)$ where $s_i=\I[i\in K']$, and $\I[\cdot]$ is the 0-1 valued indicator function, denote the joint action where agents in set $K'\subseteq[n]$ choose repair. Clearly,  $\phi(s_{K})=1$. If $i\in K$, agent $i$ does not reduce cost by switching from RE to DN since $K$ is a minimal cover therefore not repairing component $i$ causes the system to fail. As noted above, switching from DN to RE never improves an agent's cost in this game. 
Further, if $K$ is the set of agents corresponding to a non-minimal vertex cover, then there must be some agent that can reduce its cost by switching from RE to DN (and system continues to function). 
Finally, if  $K' \ne \emptyset$ is the set of agents with one or more agents short of a vertex cover, then any agent in $K'$  can reduce its cost by switching from RE to DN. This establishes that besides $0^n$, only possible NE must correspond to a minimal vertex cover. In particular, this implies that $\OPT=k^*$, where $k^*$ is the size of the smallest vertex cover $K^*$ of $\cG$, and the corresponding NE is $s_{K^*}$. 

To complete the reduction, we consider the game defined above with subsidy budget $n^*=k$. We will show a bijection between the YES and NO instances of the two decision problems to complete the proof.

If there exists a vertex cover of size $k$, then the smallest vertex cover $K^*$ has size $k^*\le k$. We design a subsidy scheme with  subsidy allocated to $k^*\le n^*$ agents, allocating subsidy of $1+\frac{1}{2n}$ for repair (the total subsidy is no more than $k^*+\frac{1}{2}$) to all agents in the minimum cover $K^*$, a  and subsidy of 0 otherwise. 
As argued above, the only candidate NE without subsidy are $0^n$ and $s_K$ corresponding to some minimal vertex cover $K$. The social cost for $0^n$ is $n$ (except the trivial case $k^*=0$) and that for $s_{K^*}$ is $k^*$ which is smaller. If we provide subsidy in our scheme $\bbS$ to the agents in $K^*$ then $0^n$ is no longer an NE. In particular, every subsidized agent in $K^*$ would now always choose repair at subsidized cost $-\frac{1}{2n}$ over doing nothing (even when the system stays broken after the repair). Thus, the social cost plus subsidy is $k^*(-\frac{1}{2n})+k^*(1+\frac{1}{2n})=k^*$, and the price of anarchy for the subsidy scheme is 1 (any agent outside of $K^*$ will prefer to do nothing to reduce their cost).

On the other hand, suppose that $\cG$ has no vertex cover of size $k$. Any minimal vertex cover of $\cG$ therefore has size at least $k+1$. Suppose $\bbS$ is a subsidy scheme with  subsidy allocated to most $k$ agents. We will show that PoA$(\bbS)>1$. Indeed, let $K$ be an arbitrary minimal  vertex cover of $\cG$. By pigeonhole principle, at least one agent in $K$ does not receive subsidy. Let $K'$ denote the (possibly empty) set of agents that receive subsidy greater than 1. As argued above, these agents will always prefer the repair action. Thus, $s_{K'}$ is a Nash equilibrium (agents without subsidy never have incentive to switch from DN to RE in this game) in the subsidized game. Note that $\phi(s_{K'})=0$ since at least one vertex is missed in any minimal vertex cover $K$, and therefore we must have an uncovered edge. 
Now $\OPT\le \text{cost}(s_{{K}})=|K|<n$. Therefore, PoA$(\bbS)> 1$ as total cost plus subsidy is at least $n$ for $s_{K'}$.
%
\end{proof}

\noindent We make a couple of remarks about the above result. Our proof implies a similar hardness result if the decision problem is stated for $\widetilde{\text{PoA}}(\bbS)$ instead of ${\text{PoA}}(\bbS)$. 
We further note that  our reduction from \textsc{Vertex-Cover} does not involve any negative literals in the boolean system function $\phi$, i.e.\ applies to monotone boolean functions where  if a component is repaired from broken to a working state, then it cannot cause the overall system to go from working to broken. {In fact, for monotone $\phi$ we can show an even stronger hardness result using parameterized complexity theory. Namely \textsc{CMG-PoAS} is W[2]-hard, by reduction from the \textsc{Dominating-Set} problem, i.e.\ deciding whether given a graph $\cG$ and integer $k$, does there exist a subset $X$ of vertices of size $k$ such that each vertex is either in the subset $X$ or has an edge connecting it to a vertex in $X$. The \textsc{Dominating-Set} problem is  known to be  W[2]-complete (\cite{downey2013fundamentals}, Lemma 23.3.1). This implies that it is unlikely that the problem is even fixed-parameter tractable (FPT), i.e.\ it is plausible that there is no computationally efficient algorithm for \textsc{CMG-PoAS} even for a fixed small subsidy budget. Further details and formal proofs are located in Appendix \ref{app:fpt}.}

{\bf Cost-sharing game}. We will now give a similar result for cost-sharing games. While it is known that computing OPT is NP-Hard in cost-sharing games \cite{meir2010minimal} and approximation algorithms are known for subsidy design \cite{buchbinder2008non}, we further show that designing a  subsidy scheme to guarantee that $\text{PoA}=1$  is also NP-hard.

\textsc{CSG-PoAS}: Given a cost-sharing game $G$ and a subsidy budget  $n^*$, 
does there exist  subsidy scheme $\bbS$ with non-zero subsidy provided to  $n^*$ actions such that $\text{PoA}(\bbS)=1$?

The key idea is to construct a cost-sharing game $G$ given an instance of the minimum set cover problem, with agents corresponding to set elements, and the actions $\cA$ corresponding to subsets available for covering plus some additional actions uniquely available to each agent. When there exists a set cover of size $k$, assigning subsidies to actions in the cover guarantees that the only possible NEs correspond to assigning actions from the set cover and the subsidized cost plus the subsidy equals cost of OPT. Conversely, we use the additional actions to show that the Price of Anarchy is greater than 1 when a subsidy smaller than the size of the minimum set cover is available.

\begin{theorem}\label{thm:csg-poa-hard}
    \textsc{CSG-PoAS} is NP-Hard.
\end{theorem}
\begin{proof}
Consider the cost-sharing game $G$ with $n$ agents that correspond to elements of $\cU$ via a bijection $\zeta:\cU\rightarrow[n]$, set of actions $\cA=\cS\uplus \cT\uplus \cV$ with $\cT=\{\{1\},\dots,\{n\}\}$ and $\cV=\{\{1,2,\dots,n\}\}$ being two distinct collections of actions available uniquely to each agent and all agents simultaneously respectively, function $f:S\mapsto\{\zeta(s)\mid s\in S\}$ which assigns action $S$ to agents corresponding its elements, and cost function $c$ given by
\begin{align*}
    c(S)=\begin{cases}
        1 &\text{ if }S\in\cS\uplus \cT,\\
        n-\epsilon&\text{ if }S\in\cV.
    \end{cases}
\end{align*}
\noindent for some $\epsilon\in(0,1)$. 
We set $n^*=k$.

Given a YES instance of \textsc{Min-Set-Cover}, we show that the above contruction yields a YES instance of \textsc{CSG-PoAS}. Indeed, let $k^*$ denote the size of the smallest set cover of $(\cU,\cS)$.  In the YES instance this means $k^*\le k<n$, and we provide subsidy of value 1 to all actions corresponding the sets in the smallest set cover.  Now any assignment of the actions to agents consistent with the set cover (i.e.\ each agent is assigned an action corresponding to one of the sets in the cover that include the agent) is a Nash Equilibrium with social cost 0. This is because each agent has (subsidized) cost 0 and therefore no incentive to switch actions. Moreover, any other state is not an NE as at least one agent will have non-zero cost and would switch to an action in the cover. Therefore, social cost plus subsidy is $k^*$ for the only NE in this case. Next we argue that the social cost OPT is $k^*$ in this case, which would imply $\text{PoA}(\bbS)=1$ for the above subsidy scheme. If possible, let there be a state $s=(a_1,\dots,a_n)$ with social cost less than $k^*$. Let  $s_\cT=\{i\mid a_i\in \cT\}$ and $s_\cV=\{i\mid a_i\in \cV\}$. Now the smallest set cover of $\cU\setminus s_{\cT}\setminus s_{\cV}$ has size at least $k^*-|s_{\cT}|-|s_{\cV}|$ since each set in the optimal cover must cover at least one element. We consider cases based on $|s_{\cV}|$. If $s_{\cV}=0$, $\text{cost}(s)\ge (k^*-|s_{\cT}|) +|s_{\cT}|= k^*$, a contradiction. If $s_{\cV}=n$, then $\text{cost}(s)=n-\epsilon>k^*$, again contradicting the assumption that cost of $s$ is less than $k^*$. Else $1\le|s_{\cV}|\le n-1$. But then, $\text{cost}(s)\ge (k^*-|s_{\cT}|-|s_{\cV}|) +|s_{\cT}|+n-\epsilon\ge k^*+1-\epsilon>k^*$. Thus, $\text{cost}(\text{OPT})=k^*$ and we have $\text{PoA}(\bbS)=1$ in this case.

Conversely, consider a NO instance of \textsc{Min-Set-Cover}. The smallest set cover of $(\cU,\cS)$ has size $k^*>k$. Consider any subsidy scheme $\bbS$ assigning subsidy of value 1 to at most $k$ actions. Clearly, all the agents that have at least one of their actions subsidized will select the subsidized action in any NE. Since the smallest set cover has size greater than $k$, there exists at least one agent  with no subsidized action. Let $A\subset[n]$ denote the set of these agents. We will show the existence of two Nash equilibria with different social costs, implying that $\text{PoA}(\bbS)>1$ in this case. Consider states $s_\cT$ and $s_\cV$ for which agents in $A$ are assigned the corresponding actions from $\cT$ and $\cV$ respectively, and agents in $[n]\setminus A$ are assigned one of the subsidized actions in either case. The social cost plus subsidy $\text{cost}(s_\cT)+k=|A|+k$ for $s_\cT$ and  $\text{cost}(s_\cV)+k=n-\epsilon+k$ for $s_\cV$. Thus, $\text{PoA}(\bbS)>1$. 
\end{proof}

\subsection{Value of Information}\label{sec:hard-voi} We also show NP-hardness of the problem of determining the minimal subsidy needed to avoid negative Value of Information in the component inspection game. Formally, the decision problem is stated below.

\textsc{CIG-VoI}: Given a component inspection game (CIG) and subsidy budget $s^*$, does some subsidy scheme $\bbS$ with total subsidy at most $s^*$ guarantee that no agent has negative value of information when a single component $j$ is inspected (i.e. for fixed inspected component $j$, $\text{VoI}_{i,j}(s,\tilde{s}^{j,y_j})\ge0$ for any $s\in\cS_{NE}(\mathbb{S})$, $\tilde{s}^{j,y_j}\in\cS_{NE}^{y_j}(\mathbb{S})$, for each agent $i\in [n]$, and each posterior component state $y_j\in\{0,1\}$)?

Our proof (Appendix \ref{app:hardness}) again involves a reduction from the \textsc{Vertex-Cover} problem, but we use a different  subsidy budget and examine the non-negativity of VoI when potentially different equilibria are selected in the prior and posterior games.

\begin{theorem}\label{thm:cig-hard}
    \textsc{CIG-VoI} is NP-Hard.
\end{theorem}

We also establish a similar hardness result for designing the optimal subsidy allocation in cost-sharing games to ensure that the value of information is non-negative for all agents.

\textsc{CSG-VoI}: Given a cost sharing  game (CSG) and subsidy budget $s^*$, does some subsidy scheme $\bbS$ with total subsidy at most $s^*$ guarantee that no agent has negative value of information when a single action $j$ is inspected (i.e. for fixed inspected action $j$, $\text{VoI}_{i,j}(s,\tilde{s}^{j,c_j})\ge0$ for any $s\in\cS_{NE}(\mathbb{S})$, $\tilde{s}^{j,c_j}\in\cS_{NE}^{c_j}(\mathbb{S})$, for each agent $i\in [n]$, and any posterior action cost $c_j$)?

\noindent We have the following hardness result, again using a reduction from  the min set cover problem.

\begin{theorem}\label{thm:csg-voi-hard}
    \textsc{CSG-VoI} is NP-Hard.
\end{theorem}

\subsection{Guaranteeing that the system functions in any NE}\label{sec:hard-sys}
Here we consider a more challenging optimization objective applicable only to the component maintenance game. We seek  the optimal subsidy scheme with the least total value across all agents that receive the subsidy, and the goal of the central agent is to disburse sufficient subsidy to guarantee that the system functions in any Nash equilibrium. The decision problem in this case is stated as follows.

\textsc{CMG-System}: Given a component maintenance game (CMG) and subsidy budget $s^*$, does some subsidy scheme $\bbS$ with total subsidy at most $s^*$ guarantee that the system functions in any  NE (i.e. $\phi(s)=1$ for any $s\in\cS_{NE}(\mathbb{S})$)? 

\purple{The reduction is similar to the proof of Theorem \ref{thm:cmg-poa-hard}.} We give a Karp mapping from any vertex cover instance $\cG,k$ to a \textsc{CMG-System} instance and use a slightly different game instance for the reduction. We show that the system is guaranteed to function in any NE iff the subsidy budget is one less than the size of a smallest vertex cover. See Appendix \ref{app:hardness} for a proof of the following result.


\begin{theorem}\label{thm:cmg-hard}
    \textsc{CMG-System} is NP-Hard.
\end{theorem}

\section{Data-driven subsidy in repeated games}\label{sec:data-driven}
The above computational hardness results for optimal subsidy design under vaious objectives motivate us to consider a beyond worst-case approach to finding a good subsidy for a given game. Specifically, we will consider the data-driven algorithm design paradigm introduced by \cite{gupta2016pac}, and further studied by \cite{balcan2017learning,balcan2018dispersion}. In this framework, we will assume access to multiple games coming from  the same  domain (e.g.\ infrastructure management in similar counties) and determine a good value of subsidy for unseen game instances from the same domain. We will consider 
games drawn i.i.d.\  from an (arbitrary, unknown) game distribution, or games arriving in an online sequence. 
In the former we 
and will be interested in having a small {\it sample complexity} of the number of game samples needed to generalize well to an unseen sample from the same distribution. For the latter, we will study regret relative to the best possible subsidy scheme over the online sequence, in hindsight. We leave open questions related to computational complexity, some recently proposed techniques are applicable to our setting \cite{balcan2022faster}.

\subsection{Sample complexity for subsidy schemes}\label{sec:sample-complexity-defs}

In this section, we define the notion of {\it sample complexity} for designing sample-based subsidy schemes for cost minimization games. The sample complexity for (uniform convergence of) a given set of subsidy schemes measures how many samples are sufficient to ensure the expected social cost of any subsidy scheme in the set approximately matches its average cost over the game samples with high probability, for any given approximation and confidence level. 
In particular, if there is a subsidy scheme in the set with small social cost over a sufficiently large set of game samples, then that scheme will almost certainly have low cost in expectation over the distribution over games from which the samples are generated. Guarantees on sample complexity are a central topic in computational learning theory \cite{vapnik1971uniform,anthony1999neural}. 

\begin{definition}[Sample complexity] The sample complexity for a class $\cS$ of subsidy schemes is a function $\cN : \R_{\ge0}\times(0,1) \rightarrow  \mathbb{Z}_{\ge1}$ defined such that for any $\epsilon>0$, any $\delta \in (0,1)$, any sample size $N\in\mathbb{Z}_{\ge1}$, and any distribution $\cD$ over the games, with probability at least $1 -\delta$ over the draw of a set $S \sim \cD^N$, for any scheme $\S$ in $\cS$, the difference between the average cost of $\S$ over 
$S$ and the expected cost of $\S$ over $\cD$ is at most $\epsilon$, whenever $N\ge \cN(\epsilon,\delta)$. In other words, 
$$\Pr_{S\sim\cD^N} \left[\exists \S \in\cS \text{ s.t. } \Bigg\lvert\frac{1}{N}\sum_{v\in S} \textsc{cost}_\S (v)- \bbE [\textsc{cost}_\S (v)]\Bigg\rvert> \epsilon\right]< \delta.$$
\end{definition}

\noindent Note that the existence of a single $\S\in\cS$ that violates the $\epsilon$-approximation for its expected cost is sufficient to cause the ``failure'' event which happens with probability $\delta$. In other words, with probability $1-\delta$, all schemes $\S$ must observe {\it uniform convergence} of sample cost to expected cost (for sufficiently large sample). The $1 -\delta$ high probability condition is needed because it is always possible that (with a very small but non-zero probability) the set of samples $S$, no matter how large, is completely unrepresentative of the distribution $\cD$ over the games. 
Clearly, $\cN(\epsilon,\delta)$ should grow as $\delta$ or $\epsilon$ shrinks since we need to ensure that the difference between the average and expected cost of each subsidy in $\cS$ is at most $\cN(\epsilon,\delta)$ with probability at least $1 -\delta$. 
The sample complexity $\cN(\epsilon,\delta)$ of class $\cS$ of course also depends on the specific subsidy class $\cS$. According to classic computational learning theory, the more ``complex'' the subsidy class $\cS$ is, the more challenging it is to bound the difference between the average and expected cost of every subsidy in $\cS$, i.e. richer subsidy classes $\cS$ have larger sample camplexity $\cN(\epsilon,\delta)$. 

For an arbitrary class $\cS$, a bound on the sample complexity allows the subsidy scheme designer to relate the expected cost of a scheme in $\cS$ which achieves minimum average cost over the set of samples to the expected cost of an optimal scheme in $\cS$, using classic arguments from learning theory \cite{mohri2018foundations}. More precisely, for a set of samples $S$ from the distribution over buyers’ values, let $\hat{S}$ be the scheme in $\cS$ that minimizes average cost over the set of samples and let $\S^*$ be the scheme in $\cS$ that minimizes expected cost over $\cD$. Finally, let $P$ be the minimum cost achievable by any scheme in $\cS$ over the support of the distribution $\cD$. For any $\delta\in (0,1)$, with probability at least $1 -\delta$ over the draw of a set of $N\ge \cN(\epsilon,\delta)$ samples $S$ from $\cD$, the difference between the expected cost of $\hat{S}$ over $\cD$ and the expected cost of $\S^*$ over $\cD$ is at most $2\epsilon$. 
Therefore, so long as there is a good sample complexity $\cN(\epsilon,\delta)$ bound for subsidy scheme class $\cS$, the scheme designer can be confident that an optimal scheme over the set of observed samples competes with an optimal scheme in $\cS$.

\paragraph{Pseudo-dimension.} Pseudo-dimension \cite{pollard2012convergence} is a well-known learning theoretic measure of complexity of a class of functions (it generalizes the Vapnik-Chervonenkis  or VC dimension to real-valued functions), and is useful in obtaining bounds on the sample complexity of fitting functions from that class to given data.  
\begin{definition}[Pseudo-dimension] Let $\cH$ be a set of real valued functions from input space $\cX$. We say that
$C = (x_1, \dots, x_m)\in \cX^m$ is pseudo-shattered by $\cH$ if there exists a vector
$r = (r_1, \dots, r_m)\in\R^m$ (called ``witness”) such that for all
$b= (b_1, \dots, b_m)\in \{\pm 1\}^m $ there exists $h_b\in \cH$ such that $\text{sign}(h_b(x_i)-r_i)=b_i$. Pseudo-dimension of $\cH$  is the cardinality of the largest set
pseudo-shattered by $\cH$.
\end{definition}

\noindent For a function class with  range $[0,H]$ and pseudo-dimension $d$, a sample complexity bound of $\cN(\epsilon,\delta)=O(\frac{H^2}{\epsilon^2}(d+\log\frac{1}{\delta}))$ is well-known \cite{anthony1999neural,balcan2020data}. We conclude this section with a  useful general lemma from data-driven algorithm design, restated in the specific context of games, for giving upper bounds on the pseudo-dimension of certain loss function classes. 

\begin{lemma} (Lemma 2.3, \cite{balcan2020data})
     Suppose that for every game $G \in \G$, the objective function $L_G(\sigma) : \R\rightarrow\R$ which maps game parameter $\sigma$ (e.g.\ subsidy allocation) to the objective value (e.g.\ cost of worst-case equilibrium) is piecewise
constant with at most $N$ pieces. Then the dual class family $\{L_\sigma(\cdot):\G\rightarrow\R\mid L_\sigma(G)= L_G(\sigma)\}$ defined on games in $\G$ has pseudo-dimension $O(\log N)$. \label{lem:ddad-1}
\end{lemma}

\subsection{Sample complexity for subsidizing games drawn from a distribution}

\textbf{Learning uniform subsidies.} We start with some initial results on learning a good value of the subsidy even in the absence of considerations about value of information, or possible equilibria in posterior games. We will consider uniform subsidy $\sigma\in\R_{\ge 0}$ conditional on repair (i.e. reduces cost of repair for all agents that choose to repair).
A simple loss objective in this uniform subsidy setting is given by

$${L}_{\text{prior}}(\sigma):=  \max_{s\in\cS_{NE}(\sigma)}\text{cost}^{\sigma}(s)+n_s\sigma,$$
where $n_s$ is the number of agents that choose repair in joint state $s$, $\cS_{NE}(\sigma)$ and $\text{cost}^{\sigma}(s)$ denote the set of Nash equilibria and (respectively) the updated total cost, when a uniform  subsidy of $\sigma$ is applied. We further assume all the repair costs as well as subsidy budget is no more than $H$, i.e. $\sigma, C_i\le H$ for each $i\in[n]$. Therefore, ${L}_{\text{prior}}(\sigma)\le (2H+1)n$.

Suppose the central agent (learner) who needs to set the subsidy has repeated instances of this game (e.g. cost matrices) drawn from a distribution. Can we learn a good value of uniform subsidy $\sigma^*$, that has small expected loss over the distribution?
Our proof involves a bound on the number of critical subsidy values at which the set of Nash equilibria could possibly change by examining the subsidy values at which the preferred actions of an individual agent $i$ conditional on any fixed joint action $s_{-i}$ of the remaining agents, and use of Lemma \ref{lem:ddad-1} \cite{balcan2020data}. \purple{Add more details/sketch on tools used.}

\begin{theorem}\label{thm:spdim}
For any $\epsilon,\delta>0$ and any distribution $\cD$ over component maintenance games with $n$ agents, $O(\frac{n^2H^2}{\epsilon^2}(n+\log\frac{1}{\delta}))$ samples of the game drawn from $\cD$ are sufficient to ensure that with probability at least $1-\delta$ over the draw of the samples, the best value of uniform subsidy $\hat{\sigma^*}$ over the sample  has expected loss $L_{\text{prior}}(\hat{\sigma^*})$ that is at most $\epsilon$ larger than the expected loss of the best value of subsidy over $\cD$.
\end{theorem}
\begin{proof}
Consider any fixed component maintenance game $G$. Observe that if actions of all agents except agent $i$ i.e.\ $s_{-i}'$ is fixed ($2^{n-1}$ possibilities), then the agent $i$ will repair their component provided the  cost under subsidy $C_i-\sigma^*+1-\bbE_\theta\phi(\x'(1,s_{-i}'))$ is smaller than $1-\bbE_\theta\phi(\x'(0,s_{-i}'))$. That is we have at most $2^{n-1}$ critical values of $\sigma^*$ where the preferred action of agent $i$ may change. Over $n$ agents, we have at most $n2^{n-1}$ such points. Moreover, the loss is piecewise constant in any fixed piece.

Given the piecewise-constant structure with a bound on the total number of pieces, the sample complexity bound follows from standard learning theoretic arguments. In more detail, by Lemma \ref{lem:ddad-1}, this implies that the pseudo-dimension of the loss function class parameterized by the subsidy value is at most $O(\log(n2^{n-1}))=O(n)$ and   classic bounds \cite{anthony1999neural,balcan2020data} imply the sample complexity result.
\end{proof}

{
\begin{remark}
Note that a naive bound of $O(2^n)$ could be derived on the pseudo-dimension of any $n$ player game, where each player has $2$ possible actions. This is because there are $2^n$ distinct states and therefore at most $2^{2^n}$ possible distinct state subsets which could correspond to a Nash Equilibrium. The critical subsidy values $\sigma^*$ correspond to values at which the set of NE changes, and for any pair of state subsets exactly one could be the set of NE for all values of subsidy above (respectively below) some critical value $\sigma^*$. The loss is again piecewise constant, and by Lemma \ref{lem:ddad-1} we have that the pseudo-dimension is $O(2^n)$. The above proof makes use of cost matrix of the component maintenance game to obtain the exponentially better upper bound of $O(n)$.
\end{remark}

\noindent \textbf{Learning non-uniform subsidies.} We are further able to obtain a sample complexity bound even for the non-uniform subsidy scheme defined above, where the central agent can provide a different subsidy to each agent depending on their component cost, failure probability and how critical the component is to overall system functionality. The subsidy scheme consists of a vector of  multiple real-valued parameters, one for each agent.\looseness-1 

$${L}_{\text{prior}}(\bbS):=  \max_{s\in\cS_{NE}(\bbS)}\text{cost}^{\bbS}(s)+\text{subs}(s).$$

\noindent We assume that $\text{subs}_i(s), C_i\le H$ for each $i\in[n]$, and therefore ${L}_{\text{prior}}(\bbS)\le (2H+1)n$. Again, we are able to give a polynomial sample complexity for the number of games needed to learn a good value of subsidy with high probability over the draw of game samples coming from some fixed but unknown distribution (proof in Appendix \ref{app:data-driven}).

\begin{theorem}\label{thm:sc-nonuniform}
For any $\epsilon,\delta>0$ and any distribution $\cD$ over component maintenance games with $n$ agents, $O(\frac{n^2H^2}{\epsilon^2}(n^2+\log\frac{1}{\delta}))$ samples of the game drawn from $\cD$ are sufficient to ensure that with probability at least $1-\delta$ over the draw of the samples, the best vector of subsidies over the sample $\hat{\sigma^*}$ has expected loss $L_{\text{prior}}$ that is at most $\epsilon$ larger than the expected loss of the best vector of subsidies over $\cD$.
\end{theorem}

\noindent 
A similar sample complexity bound can also be given for learning conditional subsidies from game samples, by minimizing a loss based on the social cost in the posterior game. See Theorem \ref{thm:sc-inspection} in Appendix \ref{app:data-driven}. Note that minimization of $L_{\text{prior}}$ corresponds to minimization of PoA$(\bbS)$. To guarantee that the system functions, we can simply add a regularization term $\lambda (1-\phi(s))$, for sufficiently large $\lambda > (2H+1)n$. 

In learning terminology, our results imply a bound on the number of sample games in the ``training set'' to do well on an unseen ``test'' game instance from the same distribution. We note that optimization over the training set is still computationally hard, but we can avoid solving the hard problem over and over again for repeated test instances.

Our techniques extend beyond losses that are based on worst case equilibria in the subsidized game. We define the following loss that considers average case Nash Equilibrium, and obtain similar same complexity guarantees as in Theorem \ref{thm:sc-nonuniform} above. 

$$\tilde{L}_{\text{prior}}(\bbS):= \frac{1}{|\cS_{NE}(\bbS)|} \sum_{s\in\cS_{NE}(\bbS)}\text{cost}^{\bbS}(s)+\text{subs}(s).$$

We remark that, generally speaking, average case NE are known to be hard to analyze and give any useful guarantees for \cite{nisan2007algorithmic}. Our result below indicates the potential of data-driven algorithm design to handle such challenging objectives and obtain meaningful learning guarantees. In particular, given a sufficiently large sample of games, we can compute near-optimal subsidy schemes with high confidence for minimizing the average cost of Nash equilibria, not just the worst case Nash equilibrium.

\begin{theorem}\label{thm:sc-nonuniform-avg}
Suppose $\text{subs}_i(s), C_i\le H$ for each $i\in[n]$. For any $\epsilon,\delta>0$ and any distribution $\cD$ over component maintenance games with $n$ agents, $O(\frac{n^2H^2}{\epsilon^2}(n^2+\log\frac{1}{\delta}))$ samples of the game drawn from $\cD$ are sufficient to ensure that with probability at least $1-\delta$ over the draw of the samples, the best vector of subsidies over the sample $\hat{\sigma^*}$ has expected loss $\tilde{L}_{\text{prior}}$ that is at most $\epsilon$ larger than the expected loss $\tilde{L}_{\text{prior}}$ of the best vector of subsidies over $\cD$.
\end{theorem}

{\bf Cost-sharing games.} We first consider the problem of learning a good value of subsidy in the prior game. The subsidy scheme consists of a vector of  multiple real-valued parameters, one for each action. We define the loss of the central agent in a game $G$ as the total social cost in the worst-case Nash Equilibrium under the subsidy scheme $\bbS$, plus the total subsidy paid out by the central agent in the scheme $\bbS$, i.e.,

$${L}_{\text{prior}}(\bbS):=  \max_{s\in\cS_{NE}(\bbS)}\text{cost}^{\bbS}(s)+\text{subs}(s).$$

 \noindent We assume that $c^\bbS(a), c(a)\le H$ for each $a\in\cA$, and therefore ${L}_{\text{prior}}(\bbS)\le 2H|\cA|$.

  \begin{theorem}\label{thm:sc-nonuniform-csg}
For any $\epsilon,\delta>0$ and any distribution $\cD$ over fair cost sharing games with $N$ agents and $|\cA|$ actions, $O\left(\frac{|\cA|^2H^2}{\epsilon^2}(|\cA|\log |\cA|N+\log\frac{1}{\delta})\right)$ samples of the game drawn from $\cD$ are sufficient to ensure that with probability at least $1-\delta$ over the draw of the samples, the best vector of subsidies over the sample $\hat{\sigma^*}$ has expected loss $L_{\text{prior}}$ that is at most $\epsilon$ larger than the expected loss of the best vector of subsidies over $\cD$.
\end{theorem}

    Another goal of the central agent is to avoid an increase in the social cost due to knowledge of the true costs. We denote the true costs by $\tilde{c}:\cA\rightarrow\R_{\ge0}$, which are assumed to be known in the posterior game, and the corresponding social costs by $\widetilde{\text{cost}}(\cdot)$. Let $\tilde{\cS}_{\text{NE}}$ denote the corresponding set of Nash Equilibria. We have

    $${L}_{\text{VoI}}(\bbS):=  \max_{s\in\tilde{\cS}_{NE}(\bbS)}\widetilde{\text{cost}}^{\bbS}(s) -  \max_{s\in\cS_{NE}(\bbS)}\text{cost}^{\bbS}(s).$$

\noindent This corresponds to the increase in the social cost relative to the prior game, when the true costs are known. The central agent would like to minimize this ``VoI'' loss.

 \begin{theorem}\label{thm:sc-voi}
For any $\epsilon,\delta>0$ and any distribution $\cD$ over fair cost sharing games with $N$ agents and $|\cA|$ actions, $O\left(\frac{|\cA|^2H^2}{\epsilon^2}(|\cA|\log |\cA|N+\log\frac{1}{\delta})\right)$ samples of the game drawn from $\cD$ are sufficient to ensure that with probability at least $1-\delta$ over the draw of the samples, the best vector of subsidies over the sample $\hat{\sigma^*}$ has expected loss $L_{\text{VoI}}$ that is at most $\epsilon$ larger than the expected loss of the best vector of subsidies over $\cD$.
\end{theorem}


\purple{TODO: Extensions for VoI?}

\purple{add cost-sharing results}

\subsection{No-regret when subsidizing in an online sequence of games}\label{sec:data-driven-ol}
In the online setting, we receive a sequence of games at times (rounds) $t=1,\dots, T$. In each round $t$, the central agent must set a value of the (say uniform) subsidy $\sigma_t$, with potentially some feedback on previous rounds but no knowledge of the game parameters (costs/priors) of the current or future rounds. This is more pessimistic (but potentially also more realistic) than the distributional setting above. In particular, the sequence of games may be adversarially picked. The performance of the algorithm is measured by the difference in the cumulative loss for the selected subsidy values and the cumulative loss of the best fixed value of subsidy in hindsight, also known as regret ($R_T$).

$$R_T:= \sum_{t=1}^T{L}_{\text{prior}}(\sigma_t)-\min_{\sigma\in[0,H]}{L}_{\text{prior}}(\sigma).$$

\noindent ``No-regret'' corresponds to $R_T$ being sublinear in $T$, and the average regret $R_T/T$ approaches zero for large $T$ in this case. We will impose a mild assumption on the repair costs $C_i$ to obtain good results in the online setting. We will assume that the costs are not known exactly, but come from some smooth distribution. Formally, 

\begin{assumption}\label{asm1}
    We assume that the probability distributions generating the costs have $\kappa$-bounded probability density, i.e. $\max_{x\in\R}f_i(x)\le \kappa$ for some $\kappa\in\R^+$, where $f_i$ denotes the probability density function for cost $C_i$.
\end{assumption}
 The adversary designing the sequence of games may  select any bad distribution as long as it is smooth. Under this assumption, our analysis above and  tools from \cite{balcan2020semi} can be used to show that the online sequence of loss functions is {\it dispersed} \cite{balcan2018dispersion}. Dispersion, informally speaking, captures how amenable a non-Lipschitz function is to online learning. As noted in \cite{balcan2018dispersion,sharma2020learning}, dispersion is a sufficient condition for learning piecewise Lipschitz functions online, even in changing environments. A formal definition is included below.\looseness-1

\begin{definition}[\label{def:dis}Dispersion, \cite{balcan2018dispersion,balcan2020semi}]
The sequence of random loss functions $L_1, \dots,L_T$ is $\beta$-{\it dispersed} for the Lipschitz constant $\ell$ if, for all $T$ and for all $\epsilon\ge T^{-\beta}$, we have that, in expectation, at most
$\Tilde{O}(\epsilon T)$ functions (here $\tilde{O}$ suppresses dependence on quantities beside $\epsilon,T$ and $\beta$, as well as logarithmic terms)
are not $\ell$-Lipschitz for any pair of points at distance $\epsilon$ in the domain $\C$. That is, for all $T$ and  $\epsilon\ge T^{-\beta}$, 

$$\bbE\left[
\max\!_{\substack{\rho,\rho'\in\C\\||\rho-\rho'||_2\le\epsilon}}\!\big\lvert
\{ t\!\in\![T] \mid |L_t(\rho)-L_t(\rho')|>\ell||\rho\!-\!\rho'||_2\} \big\rvert \right]
\le  \Tilde{O}(\epsilon T).$$
\end{definition}

\noindent Under Assumption \ref{asm1}, we have the following guarantee about online learning of uniform subsidy in a sequence of games, namely one can predict good values of subsidy (with $\tilde{O}\left(\sqrt{\frac{n}{T}}\right)$ average expected regret over $T$ online rounds). We establish $\frac{1}{2}$-dispersion of the sequence of loss functions under the above assumption and use results from \cite{balcan2020semi} to obtain the regret bound (proof in Appendix \ref{app:data-driven}). 

\begin{theorem}\label{thm:online-regret}
Suppose Assumption \ref{asm1} holds. Let $L_1,\dots, L_T:[0,H]\rightarrow[0,(2H+1)N]$ denote an independent sequence of losses ${L}_{\text{prior}}(\sigma)$ as a function of the uniform subsidy value $\sigma$, in an online sequence of $T$ component maintenance games. 
Then the sequence of functions is $\frac{1}{2}$-dispersed and there is an online algorithm with $\Tilde{O}(\sqrt{nT})$ expected regret.
\end{theorem}

\noindent The online algorithm which achieves the above regret guarantee is the Exponential Forecaster algorithm of \cite{balcan2018dispersion}, a continuous version of the well-known exponential weights update algorithm \cite{cesa2006prediction}. We can further extend the result to learning a non-uniform subsidy scheme under Assumption \ref{asm1}, with a $\tilde{O}(\sqrt{nT})$ regret bound, using the same online algorithm. 

\begin{theorem}\label{thm:online-regret-nu}
Suppose Assumption \ref{asm1} holds. Let $L_1,\dots, L_T:[0,H]^n\rightarrow[0,(2H+1)N]$ denote an independent sequence of losses ${L}_{\text{prior}}(\bbS)$ as a function of the  subsidy scheme $\bbS$ parameterized by subsidy values $\{\sigma_i\}$, in an online sequence of $T$ component maintenance games. 
Then the sequence of functions is $\frac{1}{2}$-dispersed and there is an online algorithm with $\Tilde{O}(\sqrt{nT})$ expected regret.
\end{theorem}

\section{Discussion}

We study the problem of resource allocation for infrastructure maintenance in systems with privately owned components, and in classical cost-sharing games. The former captures the typical organization of engineering systems, computer networks, or project pipelines, and the latter corresponds to typical market systems. We identify useful metrics related to the well-studied price of anarchy, as well as the recently introduced value of information metric that a central agent may care about, and examine the challenge in optimally allocating resources to optimize these metrics. 

Our work employs a data-driven approach, which has not been previously employed in the literature of subsidy or taxation design where the focus has largely been on designing approximations for worst-case game parameters. An interesting further question is to extend this idea of analyzing ``typical'' games to potentially obtain better domain-specific subsidy schemes in other interesting games. Also our learning-based approach allows modeling more realistic settings, where the game parameters or even set of agents may change over time.

\section*{Acknowledgement}
We thank Chaochao Lin for helpful initial discussions, and are grateful to Avrim Blum, Hedyeh Beyhaghi, Siddharth Prasad, Keegan Harris and Rattana Pukdee for useful comments. This material is based on work supported by the National Science Foundation under grants CCF1910321, IIS 1901403, and SES 1919453.


%
%
%
\bibliographystyle{alpha}
\bibliography{main}
%





\appendix
\section*{Appendix}
\section{Proofs from Section \ref{sec:opt-subs-2-series}}\label{app:2-series}

\textbf{Proposition \ref{prop:poa-no-subs} (restated).} \textit{In the two-agent series prior game (defined above and cost matrix noted in the first row of Table \ref{tab:2-series}), the Price of Anarchy in the absence of subsidy is at least $\text{PoA}\ge \frac{2}{p_1+p_2}$, 
for some repair costs $C_1,C_2$. More generally, for n agents, $\text{PoA}\ge \tilde{H}/\tilde{G}^n$ for some repair costs $C_1,\dots,C_n$, where $\tilde{H}$ and $\tilde{G}$ are the harmonic and geometric means, respectively, of the prior probabilities $p_1,\dots,p_n$.
}
\begin{proof}
    We set $C_1=\overline{p_1}p_2$ and $C_2=\overline{p_2}p_1$. Observe that DN-DN is an equilibrium since\footnote{In our notation, $\overline{p_1p_2}:=1-p_1p_2, \text{ while } \overline{p_1}\cdot \overline{p_2}:=(1-p_1)\cdot(1-p_2)$.} $\overline{p_1p_2}=\overline{p_1}+\overline{p_2}p_1\le \overline{p_1}+C_2$, and similarly $\overline{p_1p_2}\le \overline{p_2}+C_1$. Also, RE-RE is an equilibrium since $C_1=\overline{p_1}p_2\le \overline{p_1}$ and $C_2=\overline{p_2}p_1\le \overline{p_2}$. Clearly,
    $$\text{PoA}\ge \frac{\text{cost}(\text{DN,DN})}{\text{cost}(\text{RE,RE})}=\frac{2\overline{p_1p_2}}{C_1+C_2}\ge\frac{2}{p_1+p_2},$$
    where the last inequality follows from the observations
    \begin{align*}
        \overline{p_1p_2}(p_1+p_2)=p_1+p_2-(p_1+p_2)p_1p_2\ge p_1+p_2-2p_1p_2 = C_1+C_2.
    \end{align*}
    For the $n$ agent series game, we set the costs $C_i=\overline{p_i}\Pi_{j\ne i}p_j$. DN$^n$ (i.e., $s=0^n$) is a Nash equilibrium, since $$C_i+\overline{\Pi_{j\ne i}p_j}=(1-p_i)\Pi_{j\ne i}p_j+1-\Pi_{j\ne i}p_j\ge \overline{\Pi_{j}p_j}.$$ Moreover, RE$^n$ is also an equilibrium as $C_i\le\overline{p_i}$ for all $i\in[n]$. Therefore,
    \begin{align*}
        \text{PoA}&\ge \frac{\text{cost}(\text{DN}^n)}{\text{cost}(\text{RE}^n)}=\frac{n\overline{\Pi_{i}p_i}}{\sum_iC_i}\\
        &=\frac{n\left(1-{\Pi_{i}p_i}\right)}{\sum_i\Pi_{j\ne i}p_j - n\Pi_{i}p_i}\\
        &\ge \frac{n\left(1-{\Pi_{i}p_i}\right)}{\sum_i\Pi_{j\ne i}p_j - (\sum_i\Pi_{j\ne i}p_j)\Pi_{i}p_i}\\
        &=\frac{n}{\sum_i\Pi_{j\ne i}p_j}\\
        &=\frac{n}{\Pi_{i}p_i\cdot\sum_i\frac{1}{p_i}},
    \end{align*}
    and the claim follows by noting $\tilde{H}=\frac{n}{\sum_i\frac{1}{p_i}}$ and $\tilde{G}=(\Pi_{i}p_i)^{1/n}$.
\end{proof}

\noindent\textbf{Theorem \ref{thm:poa-series} (restated).} \textit{
    Consider the two-agent series component maintenance game with $C_1,C_2>0$ and $0<p_1,p_2<1$. Let $s^*=\I\{(C_1,C_2)\in [\overline{p_1}p_2,\overline{p_1}]\times [\overline{p_2}p_1,\overline{p_2}]\}\cdot\min\{C_1-\overline{p_1}p_2, C_2-\overline{p_2}p_1\}$, where $\I\{\cdot\}$ denotes the 0-1 valued indicator function. Then there exists a subsidy scheme $\bbS$  with total subsidy $s$ for any $s>s^*$ such that $\widetilde{\text{PoA}}(\mathbb{S})=1$. Moreover, a total subsidy of at least $s^*$ is necessary for any subsidy scheme $\bbS$ that guarantees $\widetilde{\text{PoA}}(\mathbb{S})=1$. 
}

\begin{proof} We will characterize the set of values of $C_1,C_2$ for which there are multiple Nash equilibria and design subsidy schemes that achieve $\widetilde{\text{PoA}}(\mathbb{S})=1$. We consider the following cases.

\begin{itemize}
    \item[Case 0:] $C_1<\overline{p_1}p_2,C_2>\overline{p_2}$. In this case the only NE is RE-DN (Table \ref{tab:2-series}). Thus, PoA $=1$ even in the absense of subsidy and $s^*=0$ in this case.
    \item[Case 1:] $C_1=\overline{p_1}p_2,C_2>\overline{p_2}$. Both RE-DN and DN-DN are Nash equilibria, and $$\text{cost(RE,DN)}=C_1+2\overline{p_2}=\overline{p_1}p_2+2\overline{p_2}=2-p_1p_2-p_2<2\overline{p_1p_2}=\text{cost(DN,DN)}.$$ \noindent An arbitrarily small subsidy to agent 1 is sufficient to guarantee $\widetilde{\text{PoA}}(\mathbb{S})=1$ (therefore $s^*=0$ works) as DN-DN would no longer be a NE.
    \item[Case 2:] $C_1>\overline{p_1}p_2,C_2>\overline{p_2}$. In this case the only NE is DN-DN. Thus, PoA $=1$ even in the absense of subsidy.
    \item[Case 3:] $C_1<\overline{p_1}p_2,C_2=\overline{p_2}$. Both RE-DN and RE-RE are Nash equilibria, and $$\text{cost(RE,DN)}=C_1+2\overline{p_2}>C_1+\overline{p_2}=C_1+C_2=\text{cost(RE,RE)}.$$ \noindent An arbitrarily small subsidy to agent 2 is sufficient to guarantee $\widetilde{\text{PoA}}(\mathbb{S})=1$ (therefore $s^*=0$ works)  as RE-DN would no longer be a NE.
    \item[Case 4:] $C_1<\overline{p_1}p_2,C_2<\overline{p_2}$. In this case the only NE is RE-RE. Thus, PoA $=1$ even in the absense of subsidy.
    \item[Case 5:] $(C_1,C_2)\in [\overline{p_1}p_2,\overline{p_1}]\times [\overline{p_2}p_1,\overline{p_2}]$. Both RE-RE and DN-DN are Nash equilibria, and OPT corresponds to RE-RE. A subsidy greater than $C_1-\overline{p_1}p_2$ to agent 1, or a subsidy greater than $C_2-\overline{p_2}p_1$ to agent 2 guarantees that the only NE is RE-RE. Further, in either case $\widetilde{\text{PoA}}(\mathbb{S})=1$ as the subsidy equals the reduction in the repair cost of the respective agent.

    Further suppose a subsidy of $s^*=s^*_1+s^*_2$ is sufficient to ensure $\widetilde{\text{PoA}}(\mathbb{S})=1$ in this case. Now if subsidy to agent 1 $s_1^*\le C_1-\overline{p_1}p_2$ and subsidy to agent 2 $s_2^*\le C_2-\overline{p_2}p_1$. Then both DN-DN and RE-RE are NE and $\widetilde{\text{PoA}}(\mathbb{S})>1$ since the worst-case equilibrium (i.e.\ DN-DN) cost does not depend on the subsidy. Therefore either $s_1^*> C_1-\overline{p_1}p_2$ or $s_2^*> C_2-\overline{p_2}p_1$, establishing that a subsidy of at least $s^*$ is necessary in this case to ensure $\widetilde{\text{PoA}}(\mathbb{S})=1$.


    \item[Case 6:] Otherwise. By symmetry, the case is similar to one of  C0 through C4 with agents 1 and 2 switched. $s^*=0$ and price of anarchy of 1 is achieved by no or arbitrarily small subsidy as above.
\end{itemize}
Note that $s^*$ is non-zero only in case C5, in which case we have established both sufficiency and necessity of a total subsidy of $s^*$ to ensure $\widetilde{\text{PoA}}(\mathbb{S})=1$.
\end{proof}

\noindent\textbf{Theorem \ref{thm:opt-series} (restated).} \textit{
    Consider the two-agent series component maintenance game with $C_1,C_2>0$ and $0<p_1,p_2<1$. Define $R\subset\R^+\times \R^+$ as the set of cost pairs that satisfy $C_1\le \overline{p_1}\land C_2\le \overline{p_2} \land (C_1\le \overline{p_1}p_2\lor C_2\le \overline{p_2}p_1)$, depicted in Figure \ref{figure: proof}. Let $s^*=\min_{(x,y)\in R}||(C_1,C_2)-(x,y)||_1$, where $||\cdot||_1$ denotes the $L_1$-norm. Then there exists a subsidy scheme $\bbS$  with total subsidy $s$ for any $s>s^*$ such that the system functions in any NE. Moreover, a total subsidy of at least $s^*$ is necessary for any subsidy scheme $\bbS$  that guarantees that the system functions in any NE.
}

\begin{proof}
    (Sufficiency of $s^*$). We do this by cases on the cost vector $(C_1,C_2)$, as follows.

    \begin{itemize}
        \item[Case 0:]  $(C_1,C_2)$ lies in the interior of $R$. In this case it is easy to see that $\text{PoA}=1$ and $s^*=0$. In particular, the conditions $C_1< \overline{p_1}$, $C_2< \overline{p_2}$ rule out DN-RE, RE-DN as candidate equilibria respectively, and since $C_1< \overline{p_1}p_2$ or $ C_2< \overline{p_2}p_1$,  DN-DN  cannot be a NE either (agent 1 or agent 2 will prefer to repair).
        \item[Case 1:] $C_1\le \overline{p_1}p_2, C_2\ge \overline{p_2}$. In this case, a  subsidy more than $s^*=C_2-\overline{p_2}$ to agent 2 is sufficient to bring the cost vector to the interior of $R$.
        \item[Case 2:]$C_2\le \overline{p_2}p_1, C_1\ge \overline{p_1}$. Symmetric to C1.
        \item[Case 3:] $\overline{p_1}p_2<C_1\le \overline{p_1}, \overline{p_2}p_1< C_2\le \overline{p_2}$. A subsidy more than $\min_i\{C_i-\overline{p_i}p_{3-i}\}$ to agent $\argmin_i\{C_i-\overline{p_i}p_{3-i}\}$ is sufficient to bring the cost vector to the interior of $R$.
        \item[Case 4:] Otherwise. It is straighforward to verify by direct calculation that a subsidy of $s^*_1=\max\{C_1-\overline{p_1},0\}+\overline{p_1}\cdot \overline{p_2}$ to agent 1 and a subsidy of $s^*_2=\max\{C_2-\overline{p_2},0\}+\overline{p_1}\cdot \overline{p_2}$ to agent 2 is sufficient.
    \end{itemize}
    The necessity argument essentially follows by updating the cost matrix with conditional subsidies $(s_1,s_2)$ and noting that the system is guaranteed to function in any NE if the costs $(C_1,C_2)$ are in the interior of $R$.
\end{proof}

\noindent In addition to expected VoI, we can also ensure posterior conditioned VoI$_{i,j}$ when component $c_j$ is inspected is non-negative for each agent $i$ and each posterior $y_j$ via subsidy. \purple{Compared to expected VoI, this is a more worst-case perspective as it includes the case when the component is broken which is typically when the agents are more likely to avoid the information about the component state.} The following result gives the optimal value of subsidy to ensure this.

\begin{theorem}\label{thm:voi-series}
    Consider the two-agent series component inspection game with $C_1,C_2>0$ and $0<p_1,p_2<1$. Then \begin{itemize}
        \item[(a)] $\text{VoI}_{i,j}(s_i,\Tilde{s}^{j,1})\ge0$ for any agents $i,j$, any  prior NE $s_i$ and any posterior NE $\Tilde{s}^{j,1}$, when the inspected component $j$ is working, except when $C_{3-j}=\overline{p_{3-j}}$ and an arbitrarily small subsidy is sufficient to ensure VoI is non-negative in this case. 
        \item[(b)] Define $R_1\subset\R^+\times \R^+$ as the set of cost pairs that satisfy $(C_j\le \overline{p_j}\land C_{3-j}\le \overline{p_{3-j}}p_j) \lor (C_j\le \overline{p_j}p_{3-j})$. Let $s^*=\min_{(x,y)\in R_1}||(C_1,C_2)-(x,y)||_1$, where $||\cdot||_1$ denotes the $L_1$-norm. Then there exists a subsidy scheme $\bbS$  with total (unconditional)  subsidy $s$ for any $s>s^*$ such that $\text{VoI}_{i,j}(s_i,\Tilde{s}^{j,0})\ge0$ for any agents $i,j$, any  prior NE $s_i$ and any posterior NE $\Tilde{s}^{j,0}$, when the inspected component $j$ is broken. Moreover, a total (unconditional) subsidy of at least $s^*$ is necessary for any subsidy scheme $\bbS$  that guarantees that the system functions in any NE.
    \end{itemize}
    
\end{theorem}

\begin{figure}[t]
  \centering
  \includegraphics[width=0.35\linewidth]{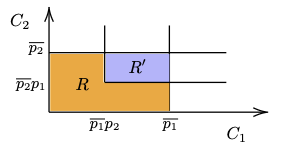}
  \caption{Cost region $R_1$ with  non-negative VoI for each agent when component 1 is inspected in a two-series game (Theorem \ref{thm:voi-series}). \label{figure: voi}}
\end{figure}

\begin{proof}
WLOG $j=1$. \begin{itemize}
    \item[(a)]
    If $y_1=1$, the only candidate posterior equilibria are DN-DN if $C_2\ge \overline{p_2}$ and DN-RE if $C_2\le \overline{p_2}$.

    If $C_2> \overline{p_2}$, then the prior NE is either DN-DN or RE-DN. In either case both agents have a non-negative Value of Information (Table \ref{tab:2-series}).
If $C_2< \overline{p_2}$, then the prior NE is either DN-DN, RE-RE or DN-RE. In each case both agents can be readily verified to have a non-negative Value of Information. If $C_2=\overline{p_2}$ and $C_1\le \overline{p_2}$, then VoI can be negative for agent 1 if the prior equilibrium is RE-RE and the posterior equilibrium is DN-DN. But in this case a (unconditional, or conditional on inspection) subsidy of $s^*_2>0$ to agent 2 ensures that DN-DN is not a posterior equilibrium and both agents have non-negative VoI. 
    \item[(b)] For $y_1=0$, we consider the following cases.
    
\begin{itemize}
    \item[Case 0:] $C_1\le \overline{p_1}p_2,C_2>\overline{p_2}$. If $C_1<\overline{p_1}p_2$, the only prior NE is RE-DN and the only posterior NE is also RE-DN. Thus, VoI $=0$ for each agent even in the absense of subsidy and $s^*=0$ in this case. If $C_1=\overline{p_1}p_2$, DN-DN is also a prior NE but VoI is still non-negative for each agent, since the posterior costs
    $$\text{cost}_1\text{(RE,DN)}=C_1+\overline{p_2}\le \overline{p_1}p_2+\overline{p_2}=\overline{p_1p_2},$$
    and 
    $$\text{cost}_2\text{(RE,DN)}=\overline{p_2}<\overline{p_1p_2}.$$
    \item[Case 1:] $\overline{p_1}p_2< C_1\le p_2,C_2>\overline{p_2}$. DN-DN is the only prior NE and the only posterior equilibrium is RE-DN (except if $C_1=p_2$, when DN-DN is also a posterior NE), and $$\text{cost}_1\text{(RE,DN)}=C_1+\overline{p_2}> \overline{p_1p_2}.$$ \noindent A subsidy of $C_1-\overline{p_1}p_2$ to agent 1 is sufficient to ensure VoI is non-negative for agent 1 as noted in  case C0. Alternatively, a subsidy of $C_2-p_1\overline{p_2}$ to agent 2  ensures that RE-RE is the only prior and posterior NE, and VoI is non-negative for each agent. The smaller of the two subsidies works and is necessary in this case.

    \purple{If we require the subsidy to be conditional on inspection, DN-DN remains the only prior NE as prior costs are not updated by the subsidy. A subsidy of $C_1-\overline{p_1}p_2$ to agent 1 works as above. In the alternative case, if we provide a smaller subsidy of $C_2-\overline{p_2}$ to agent 2, the posterior NE is switched to RE-RE but the prior remains  DN-DN and a 
    subsidy of $C_1-\overline{p_1p_2}$ to agent 1 is still required to ensure non-negative VoI for agent 1.}
    
    \item[Case 2:] $C_1>p_2,C_2>\overline{p_2}$. In this case the only NE is DN-DN in both prior and posterior games. The value of information is negative for both agents in the absense of subsidy as the posterior costs are 1. Subsidy schemes described in case C1 above can be verified to be optimal in this case as well.
    \item[Case 3:] $C_1<\overline{p_1}p_2,C_2=\overline{p_2}$. Both RE-DN and RE-RE are prior and posterior Nash equilibria, and $$\text{cost(RE,DN)}=C_1+2\overline{p_2}>C_1+\overline{p_2}=C_1+C_2=\text{cost(RE,RE)}.$$ \noindent An arbitrarily small subsidy to agent 2 is sufficient to guarantee non-negative VoI  as RE-DN would no longer be a NE.



    \item[Case 4:] Otherwise. Value of Information for agent 1 is negative since either DN-DN or DN-RE is a posterior NE with cost 1, and there is a prior NE with smaller cost to agent 1. It may be verified that the subsidy scheme from Theorem \ref{thm:opt-series} is optimal in ensuring non-negative VoI in this case. 
\end{itemize}%
\end{itemize}
\end{proof}

\section{Optimal subsidy in two agent parallel game}\label{app:2-parallel}

The cost matrix for a two-agent parallel game is summarized in Table \ref{tab:2-parallel}. Here we consider a slight generalization where there are two agents with reparable components connected in parallel, and there are additional components connected in parallel which are not maintained by any agent. Let $p_R$ denote the probability that remaining components work. 
\subsection{Guaranteeing system functions in any NE}

\begin{theorem}\label{thm:opt-parallel}
    Consider the two-agent parallel component maintenance game such that $C_1,C_2>0$ and $0<p_1,p_2<1$. 
    Let $s^*=\I\{(C_1,C_2)\in [\overline{p_R}\cdot\overline{p_1}\cdot\overline{p_2},\infty)^2\}\cdot(\min\{C_1,C_2\}-\overline{p_R}\cdot\overline{p_1}\cdot\overline{p_2})$, where $\I\{\cdot\}$ denotes the 0-1 valued indicator function. Then there exists a subsidy scheme $\bbS$  with total subsidy $s$ for any $s>s^*$ such that the system functions in any NE. Moreover, a total subsidy of at least $s^*$ is necessary for any subsidy scheme $\bbS$  that guarantees that the system functions in any NE.
\end{theorem}

\begin{proof}
    Note that the system does not function only in DN-DN, for which to be a NE we must have $(C_1,C_2)\in [\overline{p_R}\cdot\overline{p_1}\cdot\overline{p_2},\infty)^2$. Now a subsidy $s^*_i>C_i-\overline{p_R}\cdot\overline{p_1}\cdot\overline{p_2}$ to agent $i$ guarantees that agent $i$ chooses repair, establishing the first part of the theorem.

    Conversely, suppose $(C_1,C_2)\in [\overline{p_R}\cdot\overline{p_1}\cdot\overline{p_2},\infty)^2$ and a subsidy scheme with total subsidy $s^*$ guarantees that the system functions in any NE. But if $s^*_i\le C_i-\overline{p_R}\cdot\overline{p_1}\cdot\overline{p_2}$ for $i\in\{1,2\}$, then DN-DN is still an NE, and the system does not function for this NE. By contradiction, the total subsidy must be at least $s^*$ as claimed. 
\end{proof}

\subsection{Value of Information} For this case as well, 
the Value of Information may be negative if a local equilibrium is selected for some parameter settings. In the following we will show a dichotomy---if the repair costs are small then  a central agent using subsidy must subsidize the full costs of repair to avoid negative Value of Information of component inspection for the agents. Otherwise, the central agent can partially subsidize to avoid negative VoI.


\begin{theorem}\label{thm:parallel-2-agents}
Suppose $C_1,C_2\notin\{0,\overline{p_R}\cdot\overline{p_2},\overline{p_R}\cdot\overline{p_1}\cdot\overline{p_2}\}$. A subsidy scheme with $s_1^*>\max\{C_1-\overline{p_R}\cdot\overline{p_1}\cdot\overline{p_2},\min\{C_1,\overline{p_R}\cdot\overline{p_2}\}\}$ and $s_2^*>\max\{C_2-\overline{p_R}\cdot\overline{p_1}\cdot\overline{p_2},\min\{C_2,\overline{p_R}\cdot\overline{p_2}\}\}$, conditional on inspection, is sufficient to avoid negative VoI for both agents when component 1 is inspected.
\end{theorem}
\begin{proof}
Note that RE-RE cannot be an equilibrium since $C_1,C_2>0$. Also negative VoI is not possible when the posterior is $y_1=1$ as the only Nash equilibrium is DN-DN with zero cost for each agent. 

First suppose $\min\{C_1,C_2\}\ge \overline{p_R}\cdot\overline{p_1}\cdot\overline{p_2}$. This implies DN-DN is the prior equilibrium. In this case, the conditional subsidies of $s_i^*=C_i-\overline{p_R}\cdot\overline{p_1}\cdot\overline{p_2}$ are sufficient to ensure that posterior repair costs for each agent is less than $\overline{p_R}\cdot\overline{p_1}\cdot\overline{p_2}\le \overline{p_R}\cdot\overline{p_2}$. Thus, DN-DN cannot be a posterior NE for $y_1=0$. For DN-RE and RE-DN, the subsidy ensures that VoI is non-negative for both agents.

\begin{table}[t]
\centering
\begin{tabular}{@{}lcccc@{}}
\toprule
Conditions & DN-DN                     & DN-RE                   & RE-DN                   & RE-RE                 \\ \midrule
Prior      & $\overline{p_R}\cdot\overline{p_1}\cdot\overline{p_2},\overline{p_R}\cdot\overline{p_1}\cdot\overline{p_2}$ & $0,C_2$       & $C_1,0$ & $C_1,C_2$ \\
$y_1=1$    & $0,0$       & $0,C_2$       & $C_1,0$ & $C_1,C_2$ \\
$y_1=0$    & $\overline{p_R}\cdot\overline{p_2},\overline{p_R}\cdot\overline{p_2}$                     & $0,C_2$       & $C_1,0$ & $C_1,C_2$ \\ \bottomrule
\end{tabular}
\vspace*{2mm}
\caption{Matrix for cost-pairs (agent $1$, agent $2$) when component $c_1$ is inspected for the two-agent parallel system.}
\label{tab:2-parallel}
\end{table}

Otherwise, we have three cases to consider w.r.t. relative choice values of repair costs and failure probabilities,

\begin{itemize}
    \item[Case 0:] $C_1< \overline{p_R}\cdot\overline{p_1}\cdot\overline{p_2} \le C_2$. In this case, RE-DN is the prior equilibrium. 
    Moreover, since $C_1< \overline{p_R}\cdot\overline{p_1}\cdot\overline{p_2}\le \overline{p_R}\cdot\overline{p_2}$, agent 1 would prefer action RE over DN and so DN-DN cannot be a posterior NE for $y_1=0$. If RE-DN is the posterior NE, then by Table \ref{tab:2-parallel} clearly VoI is non-negative for both agents. So it only remains to consider the posterior equilibrium DN-RE. If $C_2>\overline{p_R}\cdot\overline{p_2}$, DN-RE cannot be an equilibrium, and we are done. If $C_2\le \overline{p_R}\cdot\overline{p_2}$, then the subsidy $s_2^*>\min\{C_2, \overline{p_R}\cdot\overline{p_2}\}$ ensures that VoI of agent 2 is non-negative even if DN-RE is the posterior equilibrium.  
    \item[Case 1:] $C_2< \overline{p_R}\cdot\overline{p_1}\cdot\overline{p_2} \le C_1$. The argument for this case is symmetric to the previous case, with DN-RE as the only possible prior equilibrium.
    \item[Case 2:] $\max\{C_1,C_2\}< \overline{p_R}\cdot\overline{p_1}\cdot\overline{p_2} $. In this case DN-RE as well as RE-DN can be prior Nash equilibria. Since $\max\{C_1,C_2\}< \overline{p_R}\cdot\overline{p_1}\cdot\overline{p_2} \le \overline{p_R}\cdot\overline{p_2}$, DN-RE and RE-DN are the only candidate posterior NE. The setting of subsidies $s_i^*>\overline{p_R}\cdot\overline{p_2}\ge \overline{p_R}\cdot\overline{p_1}\cdot\overline{p_2} > \max\{C_1,C_2\}$ ensures that VoI is non-negative for both agents in this case.  
\end{itemize}
\end{proof}

\section{Fixed parameter intractability of designing subsidy to minimize the Price of Anarchy}\label{app:fpt}

We  strengthen our hardness results from Section \ref{sec:hardness} here, by showing that it is unlikely that the optimal subsidy design problem is even {\it fixed parameter tractable} with the subsidy budget as the fixed parameter. Formally, a problem is fixed parameter tractable (FPT) with respect to parameter $k\in\bbZ^+$ if there exists an algorithm running in $f(k)\times n^{O(1)}$ time, where $f$ is a function of $k$ which is independent of the input size $n$. The {\it W hierarchy} is a sequence of computational complexity classes which, roughly speaking, indicate fixed parameter intractability in an increasing order of conjectured hardness. A parameterized problem $L$ is in the class W$[i]$, if every instance 
$(x,k)$ can be transformed in FPT time to a combinatorial circuit that has {\it weft} at most $i$, such that 
$(x,k)\in L$ if and only if there is a satisfying assignment to the inputs that assigns {\it true} to exactly $k$ inputs. The weft is the largest number of logical units with fan-in greater than two on any path from an input to the output. Note that W$[i]\subseteq$W[$j$] for each $i\le j \in\bbZ_{\ge 0}$, and the inclusion is conjectured to be strict. We refer the reader to standard texts for further details on parameterized complexity including the W hierarchy \cite{downey2013fundamentals}.

Our hardness result involves reduction from the \textsc{Dominating-Set} problem, which is known to be W[2]-complete. Formally, the problem may be stated as follows.

\textsc{Dominating-Set}: Does a given (undirected, unweighted) graph $\cG=(V,E)$ admit a dominating set\footnote{A subset $X$ of vertices is said to be a dominating set if for every vertex $v\in V\setminus X$ there is an edge $(x,v)\in E$ for some $x\in X$.} of size $k$?

We will show W[2]-hardness of the subsidy design problem \textsc{CMG-PoAS} for optimizing the Price of Anarchy, stated in Section \ref{sec:hardness}.

\begin{theorem}\label{thm:cmg-poa-w2-hard}
    \textsc{CMG-PoAS} is W[2]-Hard.
\end{theorem}
\begin{proof}
We will reduce the \textsc{Dominating-Set} problem to \textsc{CMG-PoAS}.  Given an instance $\cG,k$ of the \textsc{Dominating-Set} problem, we create a corresponding \textsc{CMG-PoAS} problem as follows. Introduce an agent $i$ for every vertex  $i\in V$ and consider the  formula ${\phi(\x)}=\bigwedge_{i\in V}(x_i\lor \bigvee_{(i,j)\in E} x_j)$, where the clauses correspond to each vertex and consist of component states $x_i$ for that vertex and $x_j$ for all agents corresponding to the neighbors of node $i$. Set the probability distribution $\theta$ to be the constant distribution with the entire probability mass on $0^n$ (i.e.\ all the components are guaranteed to fail without repair). Set repair cost $C_i=1$ for all components $i$. Then the cost function for agent $i$ for joint action $s=(s_i,s_{-i})$ is given by\looseness-1
\begin{align*}
    l_i(s_i,s_{-i},\theta)=\bbE_{\x\sim \theta}[\text{cost}_i]=C_is_i+P_{\phi}(\theta)=1\cdot s_i+1-\phi(\x')=s_i+1-\phi(s),
\end{align*} where $x_i'=\max\{0,s_i\}=s_i$ denotes the  component state after agent $i$ takes action $s_i$.

We  proceed to characterize the set of Nash equilibria of this game. Note that  $s=0^n$ is a NE for this game, since any repair action by any agent increases the agent's cost by $1$ if the repair does not change the state $\phi$ of the system, and by $0$ otherwise (since $\phi$ is monotonic, $\phi(s)$ can only change from 0 to 1 on repair). This also implies that no agent has any incentive to switch from DN to RE. Further note that if the components corresponding to a dominating set $X$ are have their state $x_i=1$, then $\phi(\x)=1$ as each clause of $\phi$ will have a positive literal corresponding to some node in $X$ (either the first literal is in $X$ or the node in $V\setminus X$ has some neighbor in $X$) by definition of a dominating set. Moreover, if the set of states with $x_i=1$ does not correspond to a dominating set, then there is some vertex $v\in V$ such that $x_v=0$ and $x_u=0$ for all neighbors $u$ of $v$ in $\cG$, and $\phi(\x)=0$ in this case as the clause corresponding to $v$ is not satisfied. We remark that the rest of the proof is very similar to the proof of Theorem \ref{thm:cmg-poa-hard}, and is included below for completeness.

We will now show that the remaining NEs for the game correspond to minimal dominating sets of $\cG$. 
Suppose $K\subseteq[n]$ be a set of agents for which the corresponding nodes in $\cG$ constitute a minimal dominating set. Let $s_{K'}=(s_1,\dots,s_n)$ where $s_i=\I[i\in K']$, and $\I[\cdot]$ is the 0-1 valued indicator function, denote the joint action where agents in set $K'\subseteq[n]$ choose repair. Clearly,  $\phi(s_{K})=1$. If $i\in K$, agent $i$ does not reduce cost by switching from RE to DN since $K$ is a minimal dominating set therefore not repairing component $i$ causes the system to fail. As noted above, switching from DN to RE never improves an agent's cost in this game. 
Further, if $K$ is the set of agents corresponding to a non-minimal dominating set, then there must be some agent that can reduce its cost by switching from RE to DN with system still functioning. 
Finally, if  $K' \ne \emptyset$ is the set of agents with one or more agents short of a dominating set, then any agent in $K'$  can reduce its cost by switching from RE to DN. This establishes that besides $0^n$, only possible NE must correspond to a minimal dominating set. In particular, this implies that $\OPT=k^*$, where $k^*$ is the size of the smallest dominating set $K^*$ of $\cG$, and the corresponding NE is $s_{K^*}$. 

To complete the reduction, we consider the game defined above with subsidy budget $n^*=k$. We will show a bijection between the YES and NO instances of the two decision problems to complete the proof.
If there exists a dominating set of size $k$, then the smallest dominating set $K^*$ has size $k^*\le k$. We design a subsidy scheme with  subsidy allocated to $k^*\le n^*$ agents, allocating subsidy of $1+\frac{1}{2n}$ for repair (the total subsidy is no more than $k^*+\frac{1}{2}$) to all agents in the minimum cover $K^*$, a  and subsidy of 0 otherwise. 
As argued above, the only candidate NE without subsidy are $0^n$ and $s_K$ corresponding to some minimal dominating set $K$. The social cost for $0^n$ is $n$ (except the trivial case $k^*=0$) and that for $s_{K^*}$ is $k^*$ which is smaller. If we provide subsidy in our scheme $\bbS$ to the agents in $K^*$ then $0^n$ is no longer an NE. In particular, every subsidized agent in $K^*$ would now always choose repair at subsidized cost $-\frac{1}{2n}$ over doing nothing (even when the system stays broken after the repair). Thus, the social cost plus subsidy is $k^*(-\frac{1}{2n})+k^*(1+\frac{1}{2n})=k^*$, and the price of anarchy for the subsidy scheme is 1 (any agent outside of $K^*$ will prefer to do nothing to reduce their cost).
On the other hand, suppose that $\cG$ has no dominating set of size $k$. Any minimal dominating set of $\cG$ therefore has size at least $k+1$. Suppose $\bbS$ is a subsidy scheme with  subsidy allocated to most $k$ agents. We will show that PoA$(\bbS)>1$. Indeed, let $K$ be an arbitrary minimal  dominating set of $\cG$. By pigeonhole principle, at least one agent in $K$ does not receive subsidy. Let $K'$ denote the (possibly empty) set of agents that receive subsidy greater than 1. As argued above, these agents will always prefer the repair action. Thus, $s_{K'}$ is a Nash equilibrium (agents without subsidy never have incentive to switch from DN to RE in this game) in the subsidized game. Note that $\phi(s_{K'})=0$ since at least one vertex is missed in any minimal dominating set $K$, and therefore we must have an uncovered edge. 
Now $\OPT\le \text{cost}(s_{{K}})=|K|<n$. Therefore, PoA$(\bbS)> 1$ as total cost plus subsidy is at least $n$ for $s_{K'}$.
%
\end{proof}

We remark that a similar stronger hardness result can be proved for the other studied objectives for subsidy design as well, i.e.\  the decision problems \textsc{CMG-System} and \textsc{CIG-VoI} are also W[2]-hard, by adapting the proofs of Theorem \ref{thm:cmg-hard} and \ref{thm:cig-hard}
along the lines of the above.

\section{Additional proofs from Section \ref{sec:hardness}}\label{app:hardness}

\textbf{Theorem \ref{thm:cmg-hard} (restated).}\textit{
    \textsc{CMG-System} is NP-Hard.
}
\begin{proof}
We will reduce the \textsc{Vertex-Cover} problem to \textsc{CMG-System}.  Given an instance $\cG,k$ of the \textsc{Vertex Cover} problem, we create a corresponding \textsc{CMG-System} problem as follows. Introduce an agent $i$ for every vertex  $i\in V$ and consider the (2-CNF) formula ${\phi(\x)}=\bigwedge_{(i,j)\in E}(x_i\lor x_j)$, where the clauses consist of states $x_i,x_j$ for all pairs $i,j$ of agents/components corresponding to edges in $E$. Set the probability distribution $\theta$ to be the constant distribution with the entire probability mass on $0^n$ (i.e.\ all the components are guaranteed to fail without repair). Set repair cost $C_i=1-\epsilon$ for $0<\epsilon<1/n$ for all components $i$.

Observe that, 
\begin{align*}
    l_i(s_i,s_{-i},\theta)&=\bbE_{\x\sim \theta}[\text{cost}_i]\\&=C_is_i+P_{\phi}(\theta)\\&=(1-\epsilon)s_i+1-\phi(\x')\\&=(1-\epsilon)s_i+1-\phi(s),
\end{align*} where $x_i'=\max\{0,s_i\}=s_i$ denotes the state of component $i$ after agent $i$ takes action $s_i$. Note that WLOG $s=0^n$ is a NE for this game, since any repair action by any agent increases the agent's cost by $(1-\epsilon)$ if the repair does not change the state $\phi$ of the system, and $0$ otherwise (since $\phi$ is monotonic, it can only change from 0 to 1). We will now show that the remaining NEs for the game correspond to vertex covers of $\cG$. 

Suppose $K\subseteq[n]$ be a set of agents for which the corresponding nodes in $\cG$ constitute a minimal vertex cover. Let $s_{K'}=(s_1,\dots,s_n)$ where $s_i=\I[i\in K']$, where $\I[\cdot]$ is the 0-1 valued indicator function and $K'\subset K$ with $|K'|=|K|-1$. Similarly, $s_{K}:=(s_1,\dots,s_n)$ where $s_i=\I[i\in K]$. Clearly, $\phi(s_{K'})=0$ and $\phi(s_{K})=1$. If $i\in K$, agent $i$ does not reduce cost by switching from RE to DN since $K$ is a minimal cover therefore not repairing component $i$ causes the system to fail. If $j\notin K$, agent $j$ does not reduce cost by switching from DN to RE as the system was already functioning. Further, if $K$ is the set of agents corresponding to a non-minimal vertex cover, then there must be some agent that can reduce its cost by switching from RE to DN. 
Finally, if  $K' \ne \emptyset$ is the set of agents with one or more agents short of a vertex cover, then any agent in $K'$  can reduce its cost by switching from RE to DN. This establishes that besides $0^n$, only possible NE must correspond to a minimal vertex cover. 

To complete the reduction, we consider the game defined above and subsidy budget $s^*=k-1$. If there exists a vertex cover of size $k$, then a minimal cover $K'$ has size $k'\le k$. We design a subsidy scheme with total subsidy $s'=k'-1\le s^*$, allocating subsidy of 1 for repair to all but one agent in the minimal cover $K'$ and subsidy of 0 otherwise. Clearly the only agent not given subsidy will choose repair since cost of repair $1-\epsilon$ is more than compensated by the change due to system state. As argued above, the only candidate NE are $0^n$ and $s_K$ corresponding to some minimal vertex cover $K$. Without the subsidy, the social cost for $0^n$ is $n$ (except the trivial case $k=0$) and that for $s_{K'}$ is $k'(1-\epsilon)$ which is smaller. If we provide subsidy in our scheme $\bbS$ to the agents in $K'$ except one then $0^n$ is no longer an NE. In particular, every subsidized agent in $K'$ would now choose repair at cost $1-\epsilon$ over doing nothing (even when the system stays broken after the repair) and the remaining agent in $K'$ will choose repair if the system is broken. Thus, the system functions in all NEs.

Conversely, suppose there exists a subsidy scheme $\bbS$ with total subsidy at most $s^*=k-1$, such that the system functions in any  NE. Then either the system functions in $0^n$ and there is a $0$-cover for graph $\cG$, or the  NE corresponds to a minimal vertex-cover $K'$ of size $k'$ as the repaired components (in the subsidized game). In the latter case, we seek to show $k'\le k$ to complete the proof. Since the system is not assumed to function for $s=s_{\kappa}$ for repair actions by agents in any $\kappa\subset K'$ with $\kappa=k'-1$, we
need to provide subsidy at least $1-\epsilon$ to all but one agent in $K'$. That is, $k-1=s^*\ge (1-\epsilon)(k'-1)>k'-1-(k'-1)/n$ (since $\epsilon<1/n$), or  $k\ge k'$ since both $k,k'$ are integers.
\end{proof}

\noindent\textbf{Theorem \ref{thm:cig-hard} (restated).} \textit{
    \textsc{CIG-VoI} is NP-Hard.
}
\begin{proof}
We will reduce \textsc{Vertex-Cover} to \textsc{CIG-VoI}. Recall that  \textsc{Vertex-Cover} is the following decision problem---given a graph $\cG=(V,E)$ and integer $k$, does there exist a vertex cover of size $k$? 
In  contrast to proof of Theorem \ref{thm:cmg-hard}, we will need to set a slightly higher subsidy and carefully adapt the argument to the value of information computation. 

We will create an instance of the \textsc{CIG-VoI} problem with $n=|V|+1$ agents, an agent each for vertices in $\cG$ and an additional agent $j=|V|+1$. The construction of the instance and several arguments are similar to the proof of Theorem \ref{thm:cmg-hard}. The key difference is that we have an additional agent $j$ that does not correspond to a vertex in $\cG$. We will consider the inspection of the component $c_j$ corresponding to this agent.  

Consider the (2-CNF) formula ${\phi(\x)}=\bigwedge_{(u,v)\in E}(x_u\lor x_v)$, where the clauses consist of states $x_u,x_v$ for all pairs $u,v$ of agents/components corresponding to edges in $E$. Set the probability distribution $\theta$ to be the constant distribution with the entire probability mass on $0^n$ (i.e.\ all the components are guaranteed to fail without repair). Set repair cost $C_i=1-\epsilon$ for $0<\epsilon<1/n$ for all components $i\in[|V|]$ and $C_j=1$. 
Therefore, 
$l_i(s_i,s_{-i},\theta)=(1-\epsilon)s_i+1-\phi(s)$ for $i\in[|V|]$ and $l_j(s_j,s_{-j},\theta)=2-\phi(s)$. Note that WLOG $s=0^n$ is a NE for this game, since any repair action by any agent increases the agent's cost by $(1-\epsilon)$ if the repair does not change the state $\phi$ of the system, and $0$ otherwise (since $\phi$ is monotonic, it can only change from 0 to 1). As shown in the proof of Theorem \ref{thm:cmg-hard}, the remaining NEs for the game correspond to minimal vertex covers of $\cG$. Moreover, since $\phi(s)$ does not depend on $s_j$ by definition, agent $j$ will always prefer action DN for any $s_{-j}$. Let $s_{K}:=(s_1,\dots,s_n)$ where $s_i=\I[i\in K]$ for any $K\subseteq V$.

Notice that the prior and posterior games (for inspection of $c_j$) have identical cost matrices and equilibria for this component inspection game. To complete the reduction, we consider the game defined above and subsidy budget $s^*=k$. Suppose there exists a vertex cover of $\cG$ of size $k$, then there exists a minimal vertex cover, say $K'$ of size $k'\le k$. We design a subsidy scheme with total subsidy $s'=k'\le s^*$, allocating subsidy of 1 for repair to exactly the agents in $K'$. Clearly, all subsidized agents will always choose repair. We claim that the only NE after subsidy is $s_{K'}$. Indeed, by the above observation, any NE must be $s_{K}$ for some $K\supseteq K'$. But if $K\ne K'$, then any agent in $K\setminus K'$ will choose to do nothing as the system would function without their repair action. Since there is exactly one NE in prior and posterior games, Value of Information is exactly zero for all agents. 

Conversely, if there is no vertex cover of size $k$, then we show that no subsidy scheme with $s^*\le k$ may guarantee that no agent has negative value of information when a single component $j$ is inspected. In this case the any vertex cover $K'$ has $|K'|>k$. We consider two cases:
\begin{itemize}
    \item[Case 0:] $|K'|>k+1$. Observe that if the subsidy provided to an agent is less than the repair cost $1-\epsilon$, then the agent will prefer to do nothing, except when repairing their component (given other players actions) changes the system state from 0 to 1. However, with a budget of $s^*=k$, the maximum number of agents that can receive a subsidy of at least $1-\epsilon$ is at most $\frac{k}{1-\epsilon}<k+1$, since $\epsilon<1/n$ and $k<n$ WLOG. Thus, at least two agents are without subsidy at least $1-\epsilon$ in $K'$, and these agents will prefer to do nothing if only the agents $K^*=\{i\in[|V|]\mid s^*_i>1-\epsilon\}$ with sufficient subsidy choose repair. Observe that both $s_{K'}$ and $s_{K^*}$ are Nash equilibria in the subsidized game. If $s_{K'}$ is chosen as the prior equilibrium and $s_{K^*}$ a posterior equilibrium, then the value of information for agents in $K'\setminus K^*$ is $(1-\epsilon)-1<0$ since the system does not work in $s_{K^*}$.
    \item[Case 1:] $|K'|=k+1$. In this case, the only new possibility is if at least $1-\epsilon$ subsidy is provided to all but one agent (say $k'$) in $K'$, then the remaining agent will choose repair. Without loss of generality, we assume $k+1<n$, and that $K'$ is a minimal vertex cover. Now if $v_{k'}$ denotes the vertex corresponding to agent $k'$ in $\cG$, and let $E'$ denote the set of edges incident on vertices  $V'\subseteq V\setminus K'$ with one end at $v_{k'}$. $E'$ is non-empty, as otherwise $K'\setminus \{v_{K'}\}$ would constitute a vertex cover for $\cG$ contradicting minimality of $K'$. Observe that $K_1=K'\setminus \{v_{K'}\} \cup V'$ is a vertex cover. Let $K_2$ denote a minimal vertex cover which is a subset of $K_1$. Now both $s_{K_2}$ and $s_{K'}$ are NEs in the subsidized game. If the former is set as the prior equilibrium, and the latter a posterior equilibrium then, the value of information is negative (equals $0-(1-\epsilon)=\epsilon-1$) for agent $k'$. 
\end{itemize}

\noindent Thus in either case, some agent has a negative value of information when the subsidy budget is $k$. This completes the reduction.
\end{proof}

\noindent\textbf{Theorem \ref{thm:csg-voi-hard} (restated).} \textit{
    \textsc{CSG-VoI} is NP-Hard.
}

\begin{proof}
Recall the minimum subset cover problem instance $(\cU,\cS,k)$ is given as follows.

\textsc{Min-Set-Cover}: Given a finite set $\cU$ of size $n$ and a collection $\cS\subseteq2^\cU$ of subsets of $\cU$, does there exist a subset $S$ of $\cS$ of size $k<n$ that covers $\cU$, i.e.\ $\cup_{S_i\in S}S_i=\cU$?

Consider the cost-sharing game $G$ with $n+1$ agents that correspond to elements of $\cU$ via a bijection $\zeta:\cU\rightarrow[n]$ plus additional agent $n+1$, set of actions $\cA=\cS\uplus \cT\uplus \cV\uplus \{\{n+1\}\}$ with $\cT=\{\{1\},\dots,\{n\}\}$ and $\cV=\{\{1,\dots,n\}\}$ being two distinct collections of actions available uniquely to each agent and $\{n+1\}$ corresponds to a unique action $a_{n+1}$ available to agent $n+1$. Function $f:S\mapsto\{\zeta(s)\mid s\in S\}$  assigns action $S$ to agents corresponding its elements, and cost function $c$ is given by
\begin{align*}
    c(S)=\begin{cases}
        1 &\text{ if }S\in\cS\uplus \cT,\\
        n-\epsilon &\text{ if }S\in\cV,\\
        \infty& \text{ if }S=\{n+1\},
    \end{cases}
\end{align*}
\noindent for $0<\epsilon<1$. 
We set $s^*=k$.

Given a YES instance of \textsc{Min-Set-Cover}, we show that the above contruction yields a YES instance of \textsc{CSG-VoI}. Let $k^*$ denote the size of the smallest set cover of $(\cU,\cS)$. In the YES instance this means $k^*\le k$, and we provide subsidy of value 1 to all actions corresponding the sets in the smallest set cover. The total subsidy used is $k^*\le k=s^*$.  Any assignment of the actions to agents consistent with the set cover is a Nash Equilibrium with social cost 0 and, any other state is not an NE as in the proof of Theorem \ref{thm:csg-poa-hard}. Revealing the cost of action $a_{n+1}$ does not impact the choices of agents $\{1,\dots,n\}$ as the action is not available to them, and agent $n+1$ either since the only available action is $a_{n+1}$. Thus, the cost of any agent in $[n]$ is 0 in any prior or posterior equilibria and the value of information is zero. The cost of agent $n+1$ can only decrease when it is revealed, and therefore VoI is non-negative for agent $n+1$ as well.


Conversely, consider a NO instance of \textsc{Min-Set-Cover}. The smallest set cover of $(\cU,\cS)$ has size $k^*>k$. Consider any subsidy scheme $\bbS$ assigning subsidy of value 1 to at most $k$ actions. All the agents that have at least one of their actions subsidized will select a subsidized action in any NE. Since the smallest set cover has size greater than $k$, there exists at least one agent  with no subsidized action. Let $A\subset[n]$ denote the set of these agents. We will show the existence of two Nash equilibria with different  costs for some agent in $A$, implying that $\text{VoI}<0$ for that agent by selecting the higher cost NE as the posterior equilibrium and the lower cost NE for the prior. Consider states $s_\cT$ and $s_\cV$ for which agents in $A$ are assigned the corresponding actions from $\cT$ and $\cV$ respectively, and agents in $[n]\setminus A$ are assigned one of the subsidized actions in either case. We have for any agent $i\in A$, $\text{cost}_i(s_\cT)=1$ but $\text{cost}_i(s_\cV)=\frac{n-\epsilon}{|A|}\ne \text{cost}_i(s_\cT)$.
\end{proof}

\section{Additional proofs from Section \ref{sec:data-driven}}\label{app:data-driven}

We include below  proof details for missing proofs for our sample complexity and online learning results.

\subsection{Sample complexity results}

\noindent\textbf{Theorem \ref{thm:sc-nonuniform} (restated).} \textit{
For any $\epsilon,\delta>0$ and any distribution $\cD$ over component maintenance games with $n$ agents, $O(\frac{n^2H^2}{\epsilon^2}(n^2+\log\frac{1}{\delta}))$ samples of the component maintenance game drawn from $\cD$ are sufficient to ensure that with probability at least $1-\delta$ over the draw of the samples, the best vector of subsidies over the sample $\hat{\sigma^*}$ has expected loss $L_{\text{prior}}$ that is at most $\epsilon$ larger than the expected loss of the best vector of subsidies over $\cD$.
}
\begin{proof}
Consider any fixed game $G$. Given any joint action $s=(s_i,s_{-i})$, an agent $i$'s decision for switching their action from $s_i$ to $\overline{s_i}:=1-s_i$ is determined by the inequality  $(C_i-\sigma^*_i)s_i+1-\Phi(s_i,s_{-i}')\le (C_i-\sigma^*_i)(\overline{s_i})+1-\Phi(\overline{s_i},s_{-i}')$, where $\Phi(s):=\bbE_\theta[\phi(\x'(s))]$, which is linear in $\sigma_i^*$, the subsidy provided to agent $i$. Thus, for each agent $i$, we have at most $2^{n-1}$ axis-parallel hyperplanes in the parameter space in $\R^n$, or a total of $n2^{n-1}$ hyperplanes overall. Moreover, the loss function as a function of the subsidy parameters is piecewise constant in any fixed piece. Therefore the loss function class is $(n,n2^{n-1})$-delineable in the sense of \cite{balcan2018general}, that is the subsidy parameter space is Euclidean in $n$ dimensions and is partitioned by at most $n2^{n-1}$ hyperplanes into regions where the loss is linear (in this case constant) in the parameters.\purple{better bound?}


\noindent By using a general result from \cite{balcan2018general} which states \purple{add stmt to appendix?} that a $(d,t)$-delineable function class has pseudo-dimension $O(d\log(dt))$, the above structural argument implies that the pseudo-dimension of the loss function class parameterized by the subsidy value is at most $O(n\log(n^22^{n-1}))=O(n^2)$ and the sample complexity result follows \cite{anthony1999neural,balcan2020data}.
\end{proof}

\noindent\textbf{Theorem \ref{thm:sc-nonuniform-avg} (restated).} \textit{
Suppose $\text{subs}_i(s), C_i\le H$ for each $i\in[n]$. For any $\epsilon,\delta>0$ and any distribution $\cD$ over component maintenance games with $n$ agents, $O(\frac{n^2H^2}{\epsilon^2}(n^2+\log\frac{1}{\delta}))$ samples of the game drawn from $\cD$ are sufficient to ensure that with probability at least $1-\delta$ over the draw of the samples, the best vector of subsidies over the sample $\hat{\sigma^*}$ has expected loss $\tilde{L}_{\text{prior}}$ that is at most $\epsilon$ larger than the expected loss $\tilde{L}_{\text{prior}}$ of the best vector of subsidies over $\cD$.
}

\begin{proof}
Our proof of Theorem \ref{thm:sc-nonuniform} above establish a piecewise constant structure for  the $L_{\text{prior}}$ loss given any fixed game $G$ and a bound on the number of hyperplanes in the subsidy parameter space that demarcate the pieces. On a fixed side of each of these $n2^{n-1}$ hyperplanes, each agent $i$ has a fixed preferred action given any $s_{-i}$, and therefore the set $\cS_{NE}(\bbS)$ of Nash equilibria is fixed over any piece. Indeed any joint action $s$ is a NE given subsidy scheme $\bbS$ if and only if $s_i$ is preferred given $s_{-i}$ for all agents $i$. Thus, $\tilde{L}_{\text{prior}}$ is also piecewise constant over the pieces induced by the same hyperplanes. Therefore, the same upper bound on the sample complexity can be obtained following the arguments in the proof of Theorem \ref{thm:sc-nonuniform}.
\end{proof}

\noindent \textbf{Learning conditional subsidies.} We will now obtain a sample complexity bound for  non-uniform subsidy schemes in component inspection games, where the central agent provides subsidy only in posterior games. Let $\bbS$ denote the subsidy scheme. Let ${\cS}_{NE}^0(\mathbb{S})$ (resp.\ ${\cS}_{NE}^1(\mathbb{S})$) denote the subset of states in $S$ corresponding to Nash equilibria when the cost for agent $i$ is the subsidized cost $\text{cost}_i^{{\bbS},0}$ (resp.\ $\text{cost}_i^{{\bbS},1}$) for posterior $y_j=0$ (resp.\ $y_j=1$). For component inspection game of component $c_1$ (wlog), define 

$${L}_{\text{posterior}}(\bbS):=  p_1{L}_{\text{posterior}}^1(\bbS)+\overline{p_1}{L}_{\text{posterior}}^0(\bbS),$$

\noindent where ${L}_{\text{posterior}}^i(\bbS):=\max_{s\in\cS_{NE}^i(\bbS)}\text{cost}^{{\bbS},i}(s)+\text{subs}^i(s)$. We assume that $\text{subs}_i^j(s)\le H, C_i\le H$ for each $i\in[n]$, $j\in\{0,1\}$, thus  ${L}_{\text{posterior}}(\bbS)\le (2H+1)n$. In this case too, we are able to give a polynomial sample complexity for the number of games needed to learn a good value of subsidy with high probability over the draw of game samples coming from some fixed but unknown distribution.

\begin{theorem}\label{thm:sc-inspection}
For any $\epsilon,\delta>0$ and any distribution $\cD$ over component inspection games with $n$ agents, $O(\frac{n^2H^2}{\epsilon^2}(n^2+\log\frac{1}{\delta}))$ samples of the component inspection game drawn from $\cD$ are sufficient to ensure that with probability at least $1-\delta$ over the draw of the samples, the best vector of subsidies over the sample $\hat{\sigma^*}$ has expected loss $L_{\text{posterior}}$ that is at most $\epsilon$ larger than the expected loss of the best vector of subsidies over $\cD$.
\end{theorem}

\begin{proof}
Consider any fixed game $G$. Given any joint action $s=(s_i,s_{-i})$, an agent $i$'s decision for switching their action from $s_i$ to $\overline{s_i}:=1-s_i$ in posterior game $y_1=y$ is determined by the inequality  $(C_i-s^y_i)s_i+1-\Phi(s_i,s_{-i}')\le (C_i-s^y_i)(\overline{s_i})+1-\Phi(\overline{s_i},s_{-i}')$ (with  $\Phi(s):=\bbE_{\Tilde{\theta}^{1,y}}[\phi(\x'(s))]$), which is linear in $s_i^y$, the subsidy provided to agent $i$ conditional on $y_1=y$. Thus, for each agent $i$, we have at most $2\cdot 2^{n-1}$ axis-parallel hyperplanes in the parameter space in $\R^{2n}$, or a total of $n2^{n}$ hyperplanes overall. Moreover, the loss function as a function of the parameters is piecewise constant in any fixed piece. Therefore the loss function class is $(2n,n2^{n})$-delineable in the sense of \cite{balcan2018general}. The rest of the argument is similar to the proof of Theorem \ref{thm:sc-nonuniform}, differing only in some multiplicative constants.
%
\end{proof}

\noindent\textbf{Theorem \ref{thm:sc-nonuniform-csg} (restated).} \textit{For any $\epsilon,\delta>0$ and any distribution $\cD$ over fair cost sharing games with $N$ agents and $|\cA|$ actions, $O\left(\frac{|\cA|^2H^2}{\epsilon^2}(|\cA|\log |\cA|N+\log\frac{1}{\delta})\right)$ samples of the game drawn from $\cD$ are sufficient to ensure that with probability at least $1-\delta$ over the draw of the samples, the best vector of subsidies over the sample $\hat{\sigma^*}$ has expected loss $L_{\text{prior}}$ that is at most $\epsilon$ larger than the expected loss of the best vector of subsidies over $\cD$.}

\begin{proof}
    Consider any fixed game $G$. Given any joint action $s=(s_i,s_{-i})$, an agent $i$'s decision for switching their action from $s_i$ to $\overline{s_i}\ne s_i$ is given as follows.

        
        Let $k=\sum_{j=1}^N\I[s_j=s_i]$ and $\overline{k}=\sum_{j=1}^N\I[s_j=\overline{s_i}]$. Clearly, $k\ge 1$ and $\overline{k}\ge 0$. Agent $i$'s decision to switch from action $s_i$ to $\overline{s_i}$ is governed by the inequality
        $$\frac{c(s_i)-c^\bbS(s_i)}{k}\lessgtr \frac{c(\overline{s_i})-c^\bbS(\overline{s_i})}{\overline{k}+1}.$$

    Thus, across all agents, we get at most $|\cA|^2N^2$ distinct hyperplanes in the subsidy parameter space corresponding to $c^{\bbS}$, corresponding to choices for $s_i,\overline{s_i},k,\overline{k}$.
    Moreover, the loss function as a function of the subsidy parameters is piecewise constant, since in any fixed piece induced by the above hyperplanes the set of NEs is fixed and the reduction in social cost $\text{cost}^{\bbS}(s)$ is exactly compensated by the increase in subsidy $\text{subs}(s)$ as the subsidy is varied within the piece. Therefore the loss function class is $(|\cA|,|\cA|^2N^2)$-delineable in the sense of \cite{balcan2018general}, that is the subsidy parameter space is Euclidean in $|\cA|$ dimensions and is partitioned by at most $|\cA|^2N^2$ hyperplanes into regions where the loss is linear (in this case constant) in the parameters.

    \noindent By using a general result from \cite{balcan2018general} which states  that a $(d,t)$-delineable function class has pseudo-dimension $O(d\log(dt))$, the above structural argument implies that the pseudo-dimension of the loss function class parameterized by the subsidy value is at most $O(|\cA|\log(|\cA|^3N^2))=O(|\cA|\log(|\cA|N))$ and the sample complexity result follows \cite{anthony1999neural,balcan2020data}.
\end{proof}

\noindent\textbf{Theorem \ref{thm:sc-voi} (restated).} \textit{
For any $\epsilon,\delta>0$ and any distribution $\cD$ over fair cost sharing games with $N$ agents and $|\cA|$ actions, $O\left(\frac{|\cA|^2H^2}{\epsilon^2}(|\cA|\log |\cA|N+\log\frac{1}{\delta})\right)$ samples of the game drawn from $\cD$ are sufficient to ensure that with probability at least $1-\delta$ over the draw of the samples, the best vector of subsidies over the sample $\hat{\sigma^*}$ has expected loss $L_{\text{VoI}}$ that is at most $\epsilon$ larger than the expected loss of the best vector of subsidies over $\cD$.}

\begin{proof}
    The key arguments are similar to the proof of Theorem \ref{thm:sc-nonuniform-csg}. We can show that the loss function class as a function of the subsidy is $(|\cA|,2|\cA|^2N^2)$-delineable and the result follows.
\end{proof}

\subsection{Online subsidy design}

A key tool is the following theorem due to \cite{balcan2020semi}. We present a simplified version (setting $M=1$ in their general result) as it will suffice for us.

\begin{theorem}[Theorem 5, \cite{balcan2020semi}] \label{thm:dispersion} Let $\cC \subset \R^d$ be contained in a ball of radius $R$ and $l_1,  \dots , l_T: \cC \rightarrow [0, H]$ be piecewise $\ell$-Lipschitz
functions that are $\frac{1}{2}$-dispersed. Then there is an online learning algorithm with regret bound $\tilde{O}(\sqrt{dT}+K_T)$, where  the soft-O notation suppresses terms in $R,H$ and logarithmic terms, provided $$\bbE\left[
\max\!_{\substack{\rho,\rho'\in\C\\||\rho-\rho'||_2\le\frac{1}{\sqrt{T}}}}\!\big\lvert
\{ t\!\in\![T] \mid |l_t(\rho)-l_t(\rho')|>\ell||\rho\!-\!\rho'||_2\} \big\rvert \right]
\le  K_T.$$
\end{theorem}

Thus, it is sufficient to establish $\frac{1}{2}$-dispersion of the sequence of loss functions, and provide a bound on the expected number of non-Lipschitz losses between worst-case pair of points in the domain, in order to establish our results in Section \ref{sec:data-driven-ol}.

\noindent\textbf{Theorem \ref{thm:online-regret} (restated).} \textit{
Suppose Assumption \ref{asm1} holds. Let $L_1,\dots, L_T:[0,H]\rightarrow[0,(2H+1)N]$ denote an independent sequence of losses as a function of the subsidy value $\sigma$, in an online sequence of $T$ component maintenance games. 
Then sequence of functions is $\frac{1}{2}$-dispersed and there is an online algorithm with $\Tilde{O}(\sqrt{nT})$ expected regret.
}

\begin{proof}
    The key idea is to observe that each loss function $L_t$ has  at most $K=n2^{n-1}$ discontinuities (as in proof of Theorem \ref{thm:spdim} above). Further, any interval of length $\epsilon$ has at most ${O}(\kappa\epsilon)$  discontinuities in that interval for any fixed loss function $L_t$, in expectation. This uses Assumption \ref{asm1}, and the observation that critical values of $\sigma^*$ are linear in some cost $C_i$. Indeed, as shown in the proof of Theorem \ref{thm:spdim}, the critical values of subsidy are given by $\sigma^*=C_i+\bbE_\theta\phi(0,s_{-i}')-\bbE_\theta\phi(1,s_{-i}')$ for some agent $i$ and joint action $s_{-i}'$.
    By Theorem 7 of \cite{balcan2020semi} then the expected number of non-Lipschitz losses on the worst interval of length $\epsilon$ is at most $K_T= \tilde{O}(T\epsilon+\sqrt{T\log(TK)})=\tilde{O}(\sqrt{(n+\log T)T})$ for $\epsilon\ge\frac{1}{\sqrt{T}}$.
    
    This implies $1/2$-dispersion of the sequence of loss functions in the sense of Definition \ref{def:dis}. 
    Now Theorem \ref{thm:dispersion} implies the desired regret bound. %
\end{proof}

\noindent\textbf{Theorem \ref{thm:online-regret-nu} (restated).} \textit{
Suppose Assumption \ref{asm1} holds. Let $L_1,\dots, L_T:[0,H]^n\rightarrow[0,(2H+1)N]$ denote an independent sequence of losses ${L}_{\text{prior}}(\bbS)$ as a function of the  subsidy scheme $\bbS$ parameterized by subsidy values $\{\sigma_i\}$, in an online sequence of $T$ component maintenance games. 
Then the sequence of functions is $\frac{1}{2}$-dispersed and there is an online algorithm with $\Tilde{O}(\sqrt{nT})$ expected regret.
}

\begin{proof}
    Each loss function $L_t$ can be partitioned by at most $K=n2^{n-1}$ axis-parallel hyperplanes into pieces such that the loss function is constant over each piece (as in proof of Theorem \ref{thm:sc-nonuniform} above). Further, the offset of each of these hyperplanes is linear in some cost $C_i$ and therefore along any $\sigma^*_i$-aligned line segment of length at most $\epsilon$, there are at most ${O}(\kappa \epsilon{T})$  functions that are non-Lipschitz on that segment, in expectation.
    Now for any pair of subsidy vectors $\sigma,\sigma'$ such that $||\sigma-\sigma'||_2\le \frac{1}{\sqrt{T}}$, we can bound the expected number of non-Lipschitz functions for which $|L_t(\sigma)-L_t(\sigma')|>0$ by taking an axis aligned path connecting $\sigma,\sigma'$ and adding up the number of non-Lipschitz functions along each segment. Suppose the segment lengths are $\epsilon_1,\dots,\epsilon_n$, then the above argument gives a bound ${O}(\kappa T \sum_i\epsilon_i)$ on the expected number of non-Lipschitz functions. By Cauchy-Schwarz inequality, we have  $\sum_i\epsilon_i\le \sqrt{n}\sqrt{\sum_i\epsilon_i^2}\le \sqrt{\frac{n}{T}}$, and the bound simplifies to ${O}(\kappa \sqrt{nT})$.
    
    
    By Theorem 4 of \cite{balcan2021data}  the expected number of non-Lipschitz losses on the worst point-pair with separation $\frac{1}{\sqrt{T}}$ is at most $K_T= {O}(\kappa \sqrt{nT}+\sqrt{T\log(TK)})=\tilde{O}(\sqrt{nT})$.
     Theorem \ref{thm:dispersion} now implies the claimed regret bound. %
\end{proof}

\end{document}